\documentclass[reqno,oneside,11pt]{amsart}

\usepackage[foot]{amsaddr}
\usepackage[utf8]{inputenc} 
\usepackage{amsmath, amssymb,bm, cases, mathtools, thmtools, bbm}
\usepackage{verbatim}
\usepackage{graphicx}\graphicspath{{figures/}}
\usepackage{multicol}
\usepackage{tabularx}
\usepackage[usenames,dvipsnames]{xcolor}
\usepackage{mathrsfs} 
\usepackage{url}

\usepackage[%
    minnames=4,maxnames=99,maxcitenames=4,
    style=alphabetic,
    doi=false,url=false,
    firstinits=true,hyperref,natbib,backend=bibtex]{biblatex}
\renewbibmacro{in:}{%
  \ifentrytype{article}{}{\printtext{\bibstring{in}\intitlepunct}}}
\bibliography{mcmc2} 

\usepackage[colorlinks,citecolor=blue,urlcolor=blue,linkcolor=RawSienna]{hyperref}
\usepackage{hypernat}
\usepackage{datetime}

\DeclareMathAlphabet\EuRoman{U}{eur}{m}{n}
\SetMathAlphabet\EuRoman{bold}{U}{eur}{b}{n}

\usepackage{euscript,microtype}

\usepackage[capitalize]{cleveref}

\crefname{lemma}{Lemma}{Lemmas}
\crefname{corollary}{Corollary}{Corollaries}
\crefname{theorem}{Theorem}{Theorems}

\makeatletter
\let\reftagform@=\tagform@
\def\tagform@#1{\maketag@@@{\ignorespaces\textcolor{gray}{(#1)}\unskip\@@italiccorr}}
\renewcommand{\eqref}[1]{\textup{\reftagform@{\ref{#1}}}}
\makeatother

\definecolor{WowColor}{rgb}{.75,0,.75}
\definecolor{SubtleColor}{rgb}{0,0,.50}

\newcounter{margincounter}

\declaretheorem[style=plain,numberwithin=section,name=Theorem]{theorem}
\declaretheorem[style=plain,sibling=theorem,name=Lemma]{lemma}
\declaretheorem[style=plain,sibling=theorem,name=Proposition]{proposition}
\declaretheorem[style=plain,sibling=theorem,name=Corollary]{corollary}

\declaretheorem[style=definition,sibling=theorem,name=Definition]{definition}
\declaretheorem[style=definition,qed=$\triangleleft$,sibling=theorem,name=Example]{example}
\declaretheorem[style=remark,qed=$\triangleleft$,sibling=theorem,name=Remark]{remark}

\numberwithin{theorem}{section}

\def\[#1\]{\begin{align}#1\end{align}}
\def\*[#1\]{\begin{align*}#1\end{align*}}

\newcommand{\bigO}{\mathcal{O}}
\newcommand{\normal}{\mathcal N}
\DeclareMathOperator*{\trace}{tr}	 

\newcommand{\Reals}{\mathbb{R}}

\newcommand{\dee}{\mathrm{d}}

\DeclareMathOperator{\ESJD}{ESJD}

\DeclareMathOperator*{\newlim}{\mathrm{lim}\vphantom{\mathrm{infsup}}}
\DeclareMathOperator*{\newmin}{\mathrm{min}\vphantom{\mathrm{infsup}}}
\DeclareMathOperator*{\newmax}{\mathrm{max}\vphantom{\mathrm{infsup}}}
\DeclareMathOperator*{\newinf}{\mathrm{inf}\vphantom{\mathrm{infsup}}}
\DeclareMathOperator*{\newsup}{\mathrm{sup}\vphantom{\mathrm{infsup}}}
\renewcommand{\lim}{\newlim}
\renewcommand{\min}{\newmin}
\renewcommand{\max}{\newmax}
\renewcommand{\inf}{\newinf}
\renewcommand{\sup}{\newsup}

\renewcommand{\Pr}{\mathbb{P}}
\def\EE{\mathbb{E}}
\DeclareMathOperator*{\var}{var}

\newtheorem*{assumption*}{Assumption}

\calclayout

\title[Optimal Scaling on General Target Distributions]{Optimal Scaling of Random-Walk
	Metropolis Algorithms on General Target Distributions\\ 
	}
\author{Jun Yang$^1$}
\email{jun@utstat.toronto.edu}
\author{Gareth O. Roberts$^2$}
\email{gareth.o.roberts@warwick.ac.uk}
\author{Jeffrey S. Rosenthal$^1$}
\email{jeff@math.toronto.edu}
\address{$^1$Department of Statistical Sciences,
	University of Toronto, Canada}
\address{$^2$Department of Statistics,
	University of Warwick, UK}

\begin{document}
	
	\maketitle
	\begin{abstract}
		{One} main limitation of the existing optimal scaling results for Metropolis--Hastings algorithms is that the assumptions on the target distribution are unrealistic. In this paper, we consider optimal scaling of random-walk Metropolis algorithms on general target distributions in high dimensions arising from {practical MCMC models from Bayesian statistics}.
		For optimal scaling by maximizing expected squared jumping distance (ESJD), we show the asymptotically optimal acceptance rate $0.234$ can be obtained under general {realistic} sufficient conditions on the target distribution. The new sufficient conditions are easy to be verified and may hold for some general classes of {MCMC models arising from Bayesian statistics applications}, which substantially generalize the product i.i.d.\ condition required in most existing literature of optimal scaling. Furthermore, we show one-dimensional diffusion limits can be obtained under slightly stronger conditions, which still allow dependent coordinates of the target distribution. We also connect the new diffusion limit results to complexity bounds of Metropolis algorithms in high dimensions.
	\end{abstract}
	
		\begin{center}
		\begin{minipage}{1\linewidth}
			\setcounter{tocdepth}{1}
			\tableofcontents
		\end{minipage}
	\end{center}

\newpage
\section{Introduction}

	Markov chain Monte Carlo (MCMC) algorithms \cite{Brooks2011,Gilks1995,Liu2008,Meyn2012,Robert2004} are now routinely used in many fields to obtain approximations of integrals that could not be tackled by common numerical methods, because of the simplicity and the scalability to high-dimensional settings. The running times of MCMC algorithms are an extremely important issue of practice.
	They have been studied from a variety of perspectives, including convergence ``diagnostics'' via the Markov chain output \cite{Gelman1992}, proving weak convergence limits of sped-up versions of the algorithms to diffusion limits \cite{Roberts1997,Roberts1998}, directly bounding the convergence in total variation distance \cite{Meyn1994,Rosenthal1995a,Rosenthal1996,Roberts1999,Jones2001,Rosenthal2002,Jones2004,Baxendale2005,Flegal2008}, {and non-asymptotic guarantees when the target distribution has a smooth and log-concave density, e.g.\ \cite{Bou-Rabee2018,Dalalyan2017,Dwivedi2018,Dalalyan2019} and the references therein}.
	
	The optimal scaling framework \cite{Roberts1997,Roberts1998,Roberts2001} is one of the most successful and practically useful ways of
	performing asymptotic analysis of MCMC methods in high-dimensions. Optimal scaling results (e.g. \cite{Christensen2005,Neal2006,Bedard2008a,Bedard2008,Neal2008,Neal2011,Neal2012,Jourdain2015,Jourdain2014,Roberts2014,Zanella2017}) facilitate optimization of MCMC performance by providing clear and
	mathematically-based guidance on how to tune the parameters defining the proposal distribution for Metropolis--Hastings algorithms \cite{Metropolis1953,Hastings1970}. For instance, classical results include  tuning the acceptance probabilities to $0.234$ for random-walk Metropolis algorithm (RWM) \cite{Roberts1997} and $0.574$ for Metropolis-adjusted Langevin algorithm (MALA) \cite{Roberts1998}. Moreover, optimal scaling results have been used to analyze and
	compare a wide variety of MCMC algorithms, such as Hamiltonian Monte Carlo (HMC) \cite{Beskos2013}, Pseudo-Marginal
	MCMC \cite{Sherlock2015}, multiple-try MCMC \cite{Bedard2012}. This yields guidance which is widely used by practitioners, especially via self-tuning or Adaptive
	MCMC methodologies \cite{Andrieu2008,Rosenthal2011}.
	
	In the original paper, \citet{Roberts1997} dealt with the RWM algorithm {starting in stationarity for target distributions which have i.i.d.\ product forms. The i.i.d.\ condition for the target and the assumption for the chain to start in stationarity are two main limitations of the optimal scaling framework. Particularly, the product i.i.d.\ condition is very restrictive.} From a practitioner’s
	perspective, target distributions of the i.i.d.\ forms are too limited a class of probability
	distributions to be useful, since they can be tackled by sampling a single
	one-dimensional target due to the product structure.
	To this day, optimal scaling results have mainly been proved for target distributions with a product structure, which severely limits their applicability.  On the other hand, practitioners use
	these tuning criteria far outside the class of target distributions of product i.i.d.\ forms. For example, extensive simulations \cite{Roberts2001,Sherlock2010} show that these optimality results also hold for more complex target distributions.


	There exists only a few extensions for correlated targets and most of them are derived for very specific models.
%
		For example, \citet{Breyer2000} studied
	target densities which are Gibbs measures and \citet{Roberts2001} studied inhomogeneous target densities. \citet{Breyer2004} studied target distributions arising
	in nonlinear regression and have a mean field structure.  	\citet{Neal2006} considered the case where updates of high-dimensional Metropolis algorithms are lower dimensional than the target density itself. 
	Later, \citet{Bedard2008} studied independent targets
	with different scales (see also \cite{Bedard2007,Bedard2008a}) and \citet{Bedard2019} studied a special family of hierarchical target distributions. \citet{Neal2008} studied spherically constrained target distributions and non-Gaussian proposals \cite{Neal2011}. 
	\citet{Sherlock2009} considered
	elliptically symmetric unimodal targets.
	\citet{Neal2012} studied densities with bounded support.  \citet{Durmus2017} considered target distributions which are differentiable in $L^p$ mean. 
	Recently, \citet{Mattingly2012} studied diffusion limits for a class of high-dimensional measures found from the approximation of measures on a Hilbert space which are absolutely continuous with respect to a Gaussian reference measure (See also \cite{Pillai2012,Beskos2009,Beskos2008,Cotter2013}). 
	{Important examples of this scenario required
		by \cite{Mattingly2012}  in uncertainty quantification problems are given in \cite{Hairer2011,Stuart2010,Chen2018}.
		However, in
		this paper we shall concentrate on the situation where absolute
		continuity with respect to a Gaussian is not a reasonable assumption,
		as is the case in many Bayesian statistics applications.
	}

		Furthermore, we do not consider the transient phase of the Metropolis--Hasting algorithms in this paper. The transient phase of high-dimensional Metropolis--Hasting algorithms are studied for example in \cite{Christensen2005,Jourdain2014,Jourdain2015,Kuntz2018,Kuntz2019}.
	{\citet{Kuntz2019} studied the RWM algorithm starting out of stationarity in the settings of \cite{Mattingly2012,Jourdain2015} when  non-product target distributions are defined in a Hilbert space being absolute continuous with respect to some Gaussian measures. Such target distributions in \cite{Kuntz2019} can arise for example in Bayesian nonparametric settings, but not in many other Bayesian statistics applications which we focus on in this paper.}

	In this paper, we consider optimal scaling of RWM algorithms on general target distributions in high dimensions arising from {practical MCMC models in Bayesian statistics}. First,
	for optimal scaling by maximizing expected squared jumping distance (ESJD), we show the asymptotically optimal acceptance rate $0.234$ can be obtained under general sufficient conditions on the target distribution. 	{Very briefly speaking, $0.234$ is asymptotically optimal if (i) each coordinate of the Markov chain is only strongly dependent with a subset of other coordinates  (see assumptions \ref{assumption_A1} and \ref{assumption_A3}); (ii) the target distribution satisfies some smoothness conditions (see assumptions \ref{assumption_A2} and \ref{assumption_A4}); (iii) as the dimension goes to infinity, a key quantity of ``roughness'' of the target concentrates to a nonzero value (see assumption \ref{assumption_A5}).}
	The new sufficient conditions are easy to check in practice and may hold for some general classes of {practical MCMC models}. 
	Our results substantially generalize the commonly used product i.i.d.\ condition. Furthermore, we show one-dimensional diffusion limits can also be obtained under relaxed conditions which still allow dependent coordinates of the target distribution. Finally, we also connect the new results of diffusion limits to complexity bounds of RMW algorithms in high dimensions. {Note that although the whole paper is focused on RWM algorithm, we believe the technical proofs in this paper can be used to relax restrictive conditions on the target distribution for more general Metropolis algorithms.}

	The paper is organized as follows. In \cref{section_background}, we give a brief background review of optimal scaling for Metropolis--Hastings algorithms and complexity bounds via diffusion limits. In \cref{section_main_results}, we present our main results, which include three parts: optimal scaling by maximizing ESJD, optimal scaling via diffusion limits, and complexity bounds via diffusion limits. In \cref{section_examples}, we demonstrate the new optimal scaling result holds for some {useful MCMC models}. In \cref{proof_key_theorem}, we prove \cref{key_theorem}, which is one of our main results. The proofs of lemmas used for proving \cref{key_theorem} and other main results, such as \cref{thm_diffusion,thm_diffusion2}, are delayed to \cref{proof_lemmas_key_theorem,proof_thm_diffusion,proof_thm_diffusion2}.
	

\section{Background on Optimal Scaling}\label{section_background}

Practical implementations of Metropolis--Hastings algorithms suffer from slow mixing for at least two reasons: the Markov chain moves very slowly to the target distribution when the proposed jumps are too short; the Markov chain stays at a state for most of the time when the proposed jumps are long but the chain ends up in low probability areas of the target distribution. The optimal scaling problem \cite{Roberts1997} considers the choice of proposed distribution to optimize mixing of the Metropolis--Hastings algorithm. We focus on one of the most popular MCMC algorithms, the RWM algorithm. This algorithm proceeds by running a Markov chain $\{X^d(t), t=0,\dots,\infty\}$ as follows. Given a target distribution $\pi^d$ on the state space $\Reals^d$ and the current state $X^d(t)=x^d$, a new state is proposed by $Y^d\sim\normal(x^d,\sigma_d^2I)$, which is sampled from a multivariate Gaussian distribution centered at $x^d$, then the proposal is accepted with probability $\min\{1, \pi^d(Y^d)/\pi^d(x^d)\}$ so that $X^d(t+1)=Y^d$. Otherwise the proposal is rejected and $X^d(t+1)=x^d$. This is precisely to ensure the Markov chain is reversible with respect to the target distribution $\pi^d$. It can be shown that {the normal proposals automatically make the RWM algorithm $\pi^d$-irreducible, aperiodic, and hence ergodic \cite{Roberts1994,Mengersen1996}.} Therefore, it will converge asymptotically to $\pi^d$ in law. 
 Note that the only computational cost involved in calculating the
 acceptance probabilities is the relative ratio of densities. Within the class of all Metropolis–Hastings algorithms, the RWM algorithm is still widely used in many applications
 because of its simplicity and robustness. 

 \subsection{Optimal Scaling via Diffusion Limits}
 The most common technique to prove optimal scaling results is to show a weak convergence to diffusion limits as the dimension of a sequence of target densities converges to infinity \cite{Roberts1997,Roberts1998}. {More specifically, even though different coordinates of the Markov chain are \emph{not} independent \emph{nor} even
 	individually Markovian, when the proposal is appropriately scaled according to the dimension, the sequence of sped-up stochastic processes formed by one fixed coordinate of each Markov chain converges to an appropriate Markovian Langevin diffusion process.} The limiting diffusion limit admits a straightforward efficiency maximization problem which leads to asymptotically optimal acceptance rate of the proposed moves for the Metropolis--Hastings algorithm.
In \cite{Roberts1997},  the target distribution $\pi^d$ is assumed to be an $d$-dimensional product density with respect to Lebesgue measure, that is
\[\label{eq_iid_product_target}
\pi^d(x^d)=\prod_{i=1}^d f(x_i),
\]
where $x^d=(x_1,x_2,\dots,x_d)$.
It is shown that with the choice of scaling $\sigma^2_d=\ell^2/(d-1)$ for some fixed $\ell>0$, individual components
of the resulting Markov chain converge to the solution of a stochastic differential
equation (SDE). More specifically, denoting $X^d=(X^d_1,X^d_2,\dots,X^d_d)$, the first coordinate of the RWM algorithm, $X^d_1$, sped up by a factor of $d$, i.e. $\{X_1^d(\lfloor d t\rfloor), t=0,1,\dots\}$, converges weakly in the usual Skorokhod topology to a limiting ergodic Langevin diffusion. 

\begin{proposition}{\cite[Theorem 1.1]{Roberts1997}}\label{lemma_RGG}
	Suppose density $f$ satisfies that $f'/f$ is Lipschitz continuous and 
	\[\label{RGG_smooth_condition}
	\int\left[\frac{f'(x)}{f(x)}\right]^8f(x)\dee x<\infty,\quad \int \left[\frac{f''(x)}{f(x)}\right]^4 f(x)\dee x<\infty.
	\]
	Then for $U^d(t):=X_1^d(\lfloor d t\rfloor)$, as $d\to\infty$, we have
	$U^d\Rightarrow U$,
	where $\Rightarrow$ denotes weak convergence in Skorokhod topology, and $U$ satisfies the following Langevin SDE
	\[\label{equ_SDE}
	\dee U(t)=(h(\ell))^{1/2}\dee B(t)+h(\ell)\frac{f'(U(t))}{2f(U(t))}\dee t,
	\]
	with $h(\ell):=2\ell^2\Phi(-\ell\sqrt{\tilde{I}}/2)$ is the speed measure for the diffusion process, $\tilde{I}:=\int \left[\frac{f'(x)}{f(x)}\right]^2f(x)\dee x$, and $\Phi$ being the standard Gaussian cumulative density function.
\end{proposition}	

This weak convergence result leads to the interpretation
that, started in stationarity and applied to target measures of the i.i.d.\ form, the
RWM algorithm will take on the order of $d$ steps to explore the invariant measure.
Furthermore, it may be shown that the value of $\ell$ which maximizes the speed measure $h(\ell)$ and,
therefore, maximizes the speed of convergence of the limiting diffusion, leads to a
universal acceptance probability, for the RWM algorithm applied
to targets of i.i.d. forms, of approximately $0.234$.
\cref{lemma_RGG} is proved in \cite{Roberts1997} using the generator approach \cite{Ethier1986}. The same method of proof has also been applied to derive optimal scaling results for other types of MCMC algorithms: for example, the convergence of MALA to diffusion limits when $\sigma^2_d=\ell^2/d^{1/3}$ (see e.g. \cite{Roberts1998,Roberts2001,Breyer2004,Christensen2005,Neal2006}) with  asymptotically optimal acceptance rate $0.574$.

\subsection{Optimal Scaling by maximizing ESJD}
Another popular technique to prove optimal scaling is by maximizing expected squared jumping distance (ESJD) \cite{Pasarica2010,Atchade2011,Roberts2014}, which is defined as follows.
\begin{definition}(Expected Squared Jumping Distance)\label{def_ESJD}
	\[
	\ESJD(d):=
	&\EE_{X^d\sim\pi^d}\EE_{Y^d}\left[\|Y^d-X^d\|^2\left(1\wedge \frac{\pi^d(Y^d)}{\pi^d(X^d)}\right)\right]
	\]
	where the expectation over $Y^d$ is taken for $Y^d\sim\mathcal{N}(x^d,\frac{\ell^2}{d-1}I)$ for given $X^d=x^d$, and $\|\cdot\|$ denotes the Euclidean distance, i.e. $\|Y^d-X^d\|^2=\sum_{i=1}^d (Y_i-X_i)^2$.
\end{definition}

	Choosing a proposal variance to maximize ESJD  is equivalent to minimizing the first-order auto-correlation of the Markov chain, and thus maximizing the efficiency if the higher order auto-correlations are monotonically increasing with respect to the first-order auto-correlation \cite{Pasarica2010}. Furthermore, if weak convergence to a diffusion limit is established, then the ESJD converges to the quadratic variance of the diffusion limit. This suggests that maximizing the ESJD is a reasonable problem.
		For example, \citet{Atchade2011} considered to maximize the ESJD to choose optimal temperature spacings for Metropolis-coupled Markov chain Monte Carlo and simulated tempering algorithms.
	Later, \citet{Roberts2014} proved a diffusion limit for the simulated tempering algorithms. Using a new comparison of asymptotic variance of diffusions, \citet{Roberts2014} showed the results in the choice of temperatures in \cite{Atchade2011} does indeed minimize the asymptotic variance of all functionals. {Another example is the optimal scaling result for HMC, with asymptotically optimal acceptance rate $0.651$ when $\sigma^2_d=\ell^2/d^{1/4}$ \cite{Beskos2013}, is proven by maximizing the ESJD.}
	
	Although establishing weak convergence of diffusion limits gives stronger guarantee than maximizing ESJD, the price to pay is to require stronger conditions on the target distribution. Maximizing ESJD instead can lead to (much) weaker conditions on the target distribution. Later in this paper, we will show that we are able to relax the restrictive product i.i.d.\ condition on the target distribution for both cases. In particular, the new sufficient conditions on the target distribution for maximizing ESJD are weak enough to allow target distributions arising from realistic MCMC models.

\subsection{Background on Complexity Bounds}

Because of the big data world, in recent years, there is much interest in the ``large $d$, large $n$'' or ``large $d$, small $n$'' high-dimensional regime, where $d$ is the number of parameters and $n$ is the sample size.  \citet{Rajaratnam2015} use the term convergence complexity to denote
the ability of a high-dimensional MCMC scheme to draw samples from the posterior, and how the ability to do so changes as the dimension of the parameter set grows. This requires the study of computer-science-style complexity bounds \cite{Cobham1965,Cook1971} in terms of running time complexity order as the ``size'' of the problem goes to infinity.
In the Markov chain context, computer scientists have been bounding convergence times of Markov chain algorithms focusing largely on spectral gap bounds for Markov chains \cite{Sinclair1989,Lovasz2003,Vempala2005,Lovasz2006,Woodard2009,Woodard2009a}. In contrast, statisticians usually study total variation distance or other metric for MCMC algorithms. 
In order to bridge the gap between statistics-style convergence bounds, and computer-science-style complexity results, in one direction, \citet{Yang2017} recently show that complexity bounds for MCMC can be obtained by quantitative bounds using a modified drift-and-minorization approach. In another direction, \citet{Roberts2016} connect existing results on diffusion limits of MCMC algorithm to the computer science notion of algorithm complexity. The main result in \cite{Roberts2016} states that any weak limit of a Markov process implies a corresponding complexity bound in an appropriate metric. More specifically, \citet{Roberts2016} connect the diffusion limits to complexity bound using the Wasserstein metric. Let $(\mathcal{X},\mathcal{F},\rho)$ be a general measurable metric space, the distance of a stochastic process $\{X(t)\}$ on $(\mathcal{X},\mathcal{F})$ to its stationary distribution $\pi$ is defined by the KR distance
\[\label{KR_distance}
\|\mathcal{L}_x(X(t))-\pi\|_{\textrm{KR}}:=\sup_{g\in \textrm{Lip}_1^1}|\EE[g(X(t))]-\pi(g)|
\]
where {$\mathcal{L}_x(X(t))$ denotes the law of $X(t)$ conditional on starting at $X(0)=x$}, $\pi(g):=\int g(x)\pi(\dee x)$ is the expected value of $g$ with respect to $\pi$, `KR' stands for `Kantorovich--Rubinstein', and $\textrm{Lip}_1^1$ is the set of all functions $g$ from $\mathcal{X}$ to $\Reals$ with Lipschitz constant no larger than $1$ and with $|g(x)|\le 1$ for all $x\in\mathcal{X}$, i.e.
\[
\textrm{Lip}_1^1:=\{g: \mathcal{X}\to\Reals, |g(x)-g(y)|\le \rho(x,y), \forall x,y\in\mathcal{X}, |g|\le 1  \}.
\]
Note that the KR distance defined in \cref{KR_distance} is exactly the $1$-st Wasserstein metric. Then it can be shown that the $\pi$-average of the KR distance to stationarity from all initial states $X(0)$ in $\mathcal{X}$ is non-increasing, which leads to the following complexity linking proposition.
\begin{proposition}{\cite[Theorem 1]{Roberts2016}}\label{lemma_RR16}
	Let $X^d=\{X^d(t), t\ge 0\}$ be a stochastic process on $(\mathcal{X},\mathcal{F},\rho)$, for each $d\in \mathbb{N}$. Suppose $X^d$ converges weakly in the Skorokhod topology as $d\to\infty$ to a c\`adl\`ag process $X^{\infty}$. Assume these processes all have the same stationary distribution $\pi$ and that $X^{\infty}$ converges weakly
	to $\pi$. Then for any $\epsilon>0$, there are $D<\infty$ and $T<\infty$ such that
	\[
	\EE_{X^d(0)\sim \pi}\|\mathcal{L}_{X^d(0)}(X^d(t))-\pi\|_{\textrm{KR}}<\epsilon,\quad \forall t\ge T, d\ge D.
	\]	
\end{proposition}
\cref{lemma_RR16} allows us to bound the convergence of the sequence of processes uniformly over all sufficiently large $d$, if the sequence of Markov processes converges weakly to a limiting ergodic process. Combining \cref{lemma_RR16} with previously-known MCMC diffusion limit results, \citet{Roberts2016} prove that the RWM algorithm in $d$ dimensions takes $\bigO(d)$ iterations to converge to stationarity. However, in \cite{Roberts2016}, the target distribution needs to be product i.i.d.\ with density satisfies all the assumptions of \cref{lemma_RGG}. Furthermore, the condition \cref{RGG_smooth_condition} is replaced by a stronger condition
\[\label{RR16_smooth_condition}
\int\left[\frac{f'(x)}{f(x)}\right]^{12}f(x)\dee x<\infty,\quad \int \left[\frac{f''(x)}{f(x)}\right]^6 f(x)\dee x<\infty.
\]

\section{Main Results}	\label{section_main_results}

In this section, we show our main results on optimal scaling of RWM algorithms on general target distributions.
We first consider optimal scaling by maximizing ESJD in \cref{subsection_ESJD}. We show asymptotic form of the ESJD in \cref{key_theorem} under very mild conditions on the target distribution. Then we show in \cref{thm_maximize_lower_bound} that if we directly maximize the asymptotic ESJD, we can obtain $0.234$ as an upper bound of the asymptotically optimal acceptance rate. Next, we show the acceptance rate $0.234$ is asymptotically optimal under one more weak law of large number (WLLN) condition on the target distribution in \cref{thm_ESJD}. {In order to give the reader a brief idea that to what extend the class of target distributions can be enlarged. We first present an example of a non-product non-i.i.d.~ class of distributions, which is a straightforward corollary of our main result in \cref{thm_ESJD}. Note that our main result includes much more general class of distributions that this simple example.} {Recall that a (probabilistic) graphical model is a family of probability distributions defined in terms of a directed or undirected graph \cite{Jordan2004}. Suppose that the statistical model can be represented as a graphical model, then we have the following corollary.}
\begin{corollary}{(A Simple Corollary of \cref{thm_ESJD})	
If the following three conditions hold, $0.234$ is indeed the asymptotic acceptance rate: (i) in the graph representation, each node of the graph has at most $o(d^{1/4})$ links; (ii) the target density $\pi^d$ is bounded and $\log\pi^d$ has up to the third bounded partial derivatives; (iii) for $X^d\sim\pi^d$,    $\frac{1}{d}\sum_{i=1}^d\left(\frac{\partial}{\partial x_i}\log \pi^d(X^d)\right)^2$ converges to a positive constant as $d\to\infty$.}
\end{corollary}
In \cref{subsection_diffusion}, we consider optimal scaling via diffusion limits. We prove the new conditions for weak convergence to diffusion limits in \cref{thm_diffusion}. We then strengthen this result to consider fixed starting state in \cref{thm_diffusion2}. 
Finally, in \cref{subsection_complexity_bounds}, we apply our new result on diffusion limits with fixed starting state to obtain complexity bounds for the RMW algorithm, which is given in \cref{corollary_complexity}.

Before presenting our main results, we first define a sequence of ``sets of typical states''. 
\begin{definition}
	We call $\{F_d\}$ a sequence of ``sets of typical states'' 
	{if $\pi^d(F_d)\to 1$.}
\end{definition}
Next, we enlarge $\{F_d\}$ in different ways, which will be used later for the new conditions on the target.
\begin{definition}
	For a given sequence of ``sets of typical states'' $\{F_d\}$, we define
	\[
	F_d^{(i)}:=\{(x_1,\dots,x_{i-1},y,x_{i+1},\dots,x_d): \exists (x_1,\dots,x_d)\in F_d, \textrm{such that } |y-x_i|<\sqrt{\log d/d} \}.
	\]
	Furthermore, we define $
	F_d^+:=\bigcup_{i=1}^d F_d^{(i)}$.
\end{definition}

\begin{remark}
It is clear from the definitions that $F_d^{(i)}$ is to enlarge the $i$-th coordinate of $x^d\in F_d$ by covering it with an open interval $(x_i-\sqrt{\log d/d}, x_i+\sqrt{\log d/d})$; $F_d^+$ is the union of $F_d^{(i)}, i=1,\dots,d$. 
Then clearly we have $F_d\subseteq F_d^{(i)}\subseteq F_d^+$. 
In practice, the difference between $F_d^+$ and $F_d$ is usually asymptotically ignorable in high dimensions.
\end{remark}

Finally, we introduce the idea of ``neighborhoods'' of a coordinate, which is later used to capture the correlation among different coordinates. We {use $\mathcal{H}_i$ to denote a collection of coordinates} which are {called ``neighborhoods'' of coordinate $i$. That is, 
 $\mathcal{H}_i\subseteq \{1,\dots,d\}$.} {We also assume $i\in \mathcal{H}_i$.}
 Although the definition of the set $\mathcal{H}_i$ is quite arbitrary, we expect that $j\in \mathcal{H}_i$ implies the coordinates $i$ and $j$ are correlated even conditional on all other coordinates. This idea of ``neighborhoods'' become clearer if the target distribution comes from a model which can be written as {a probabilistic graphical model \cite{Jordan2004}. For a graphical model, it is convenient to define the ``neighborhood'' $j\in \mathcal{H}_i$ if there is an edge between nodes $i$ and $j$. In this definition, clearly $j\notin \mathcal{H}_i$ implies that the two coordinates $i$ and $j$ are conditional independent given all the other $d-2$ coordinates.}

\subsection{Optimal Scaling for Maximizing ESJD}\label{subsection_ESJD}

Suppose $\{F_d\}$ is a sequence of ``sets of typical states'' and $\{\mathcal{H}_i\}$ are collections of ``neighborhoods'' for each coordinate. Throughout the paper, we assume $\sup_{i \in \{1,\dots,d\}}|\mathcal{H}_i|< l_d$ where $l_d=o(d)$.

\begin{remark}
	For graphical models, if we define $\mathcal{H}_i$ as the collection of nodes that is directly connected to $i$ by an edge, then $l_d=o(d)$ rules out ``dense graphs'' for which $l_d\propto d$.
\end{remark}

Now we introduce the first assumption \ref{assumption_A1} on the target $\pi^d$.
\begin{equation}
\label{assumption_A1}\tag{A1}
\begin{split}
&\sup_{(i,j): j\notin \mathcal{H}_i}\sup_{x^d\in F_d^+}\frac{\partial^2 \log \pi^d(x^d)}{\partial x_i \partial x_j}=o(1),\quad
\sup_{(i,j):j\in\mathcal{H}_i} \sup_{x^d\in F_d^+} \frac{\partial^2 \log \pi^d(x^d)}{\partial x_i \partial x_j}=o(\sqrt{d/l_d}).
\end{split}
\end{equation}
\begin{remark}\label{remark_A1}
	For graphical models, if node $i$ is not directly connected to node $j$, we always have $\frac{\partial^2 \log \pi^d(x^d)}{\partial x_i \partial x_j}=0$. Therefore, in order to make \ref{assumption_A1} hold, it suffices to check for each edge of the graph, say $(i,j)$, that $\frac{\partial^2 \log \pi^d(x^d)}{\partial x_i \partial x_j}=o(\sqrt{d/l_d})$. Since we have assumed $l_d=o(d)$, this is a very weak condition. For example, \ref{assumption_A1} holds for all graphical models with bounded second partial derivatives.
\end{remark}

Next, we {denote the conditional density of the $i$-th and $j$-th coordinates, given all the other coordinates fixed, by $\pi_{i,j|-i-j}:= \pi^d(x_i,x_j\,|x_{-i-j})$} where $x_{-i-j}$ with $i<j$ denotes all coordinates of $x^d$ other than the $i$-th, and $j$-th coordinates, i.e. $$x_{-i-j}:=(x_1, \dots, x_{i-1}, x_{i+1}, \dots, x_{j-1}, x_{j+1}, \dots, x_d).$$ 
{Note that $\pi_{i,j|-i-j}$ is a probability measure in $\Reals^2$.}
Then we introduce the next assumption \ref{assumption_A2} on the target as follows.
\[\label{assumption_A2}\tag{A2}
&\sup_{(i,j): j\notin \mathcal{H}_i}\sup_{\{x_{-i-j}: x^d\in F_d\}}\int \frac{\partial^2 \pi_{i,j|-i-j}}{\partial x_i^2} \frac{\partial^2 \pi_{i,j|-i-j}}{\partial x_j^2}\frac{1}{\pi_{i,j|-i-j}}\dee x_i\dee x_j=o(1).
\]
\begin{remark}\label{remark_A2}
	The assumption \ref{assumption_A2} is very weak, since it is only to require that the target has a ``flat tail''. To see this, consider the target distribution $\pi^d$ has the special i.i.d.\ product form of \cref{eq_iid_product_target}, then \ref{assumption_A2} reduces to
	\[
	\int \frac{\partial^2 f(x_i)f(x_j)}{\partial x_i^2}\frac{\partial^2 f(x_i)f(x_j)}{\partial x_j^2}\frac{1}{f(x_i)f(x_j)}\dee x_i \dee x_j=\left(\int \frac{\dee^2 f(x)}{\dee x^2}\dee x\right)^2=0,
	\]
	when $f$ has a ``flat tail'' so that $\frac{\dee f(x)}{\dee x}\to 0$ when $|x|\to\infty$. Similarly, for graphical models, if there is no edge between $i$ and $j$, then when $\pi^d$ has ``flat tail'' we have $\int \frac{\partial^2 \pi_{i,j|-i-j}}{\partial x_i^2} \frac{\partial^2 \pi_{i,j|-i-j}}{\partial x_j^2}\frac{1}{\pi_{i,j|-i-j}}\dee x_i\dee x_j=0$.
\end{remark}

The next assumption is about conditions on the third partial derivatives.
	\begin{equation}
\label{assumption_A3}
\tag{A3}
\begin{split}
&\sup_{(i,j): j\notin \mathcal{H}_i}\sup_{x^d \in {\Reals^d}}\frac{\partial^3 \log \pi^d(x^d)}{\partial x_i^2 \partial x_j}=o(1),\quad\sup_{(i,j): j\in \mathcal{H}_i}\sup_{x^d \in {\Reals^d}}\frac{\partial^3 \log \pi^d(x^d)}{\partial x_i^2 \partial x_j}=o(d/l_d), \\ &{\sup_i \sup_{x^d \in \Reals^d} \frac{\partial^3 \log \pi^d(x^d)}{\partial x_i^3}=o(d^{1/2}),}\quad
{\sum_{i\neq j\neq k}\left(\sup_{x^d\in {\Reals^d}}\left|\frac{\partial^3 \log \pi^d(x^d)}{\partial x_i\partial x_j \partial x_k}\right|\right)=o(d^{3/2}).}
\end{split}
\end{equation}
\begin{remark}\label{remark_A3}
	We consider graphical models that satisfy \ref{assumption_A3}. The first three equations of \ref{assumption_A3} are similar to \ref{assumption_A1} and they hold for all graphical models with bounded third partial derivatives. {Recall that, in graph theory, a $n$-clique of a graph is a fully-connected subset of nodes of the graph with cardinality $n$. The last equation of \ref{assumption_A3} then involves the number of $3$-cliques in the graph.} Note that for many realistic hierarchical models, there are no $3$-cliques for the corresponding graphs, which implies $\sum_{i\neq j\neq k}\left|\frac{\partial^3 \log \pi^d(x^d)}{\partial x_i\partial x_j \partial x_k}\right|=0$. Even for the worst case, considering a graph that has $d$ nodes and each has $l_d$ neighbors, since there are $d l_d/2$ links, the number of $3$-cliques is at most ${l_d\choose 2}d/3=\bigO(l_d^2 d)$. Therefore, \ref{assumption_A3} holds for any graphical model with $l_d={o(d^{1/4})}$ and bounded third partial derivatives.
\end{remark}	

The next assumption is the last assumption before our first main result. We first define a quantity which measures the ``roughness'' of $\log \pi^d$.
\[
I_d(x^d):=\frac{1}{d}\sum_{i=1}^d \left(\frac{\partial}{\partial x_i}\log \pi^d(x^d) \right)^2.
\]
Similarly, we can consider $I_d(X^d)$ where $X^d\sim \pi^d$ as a random variable. Later we will see that it turns out that $I_d(X^d)$ is a key quantity for optimal scaling results. Assumption \ref{assumption_A4} is as follows.

There exists $\alpha$ with $0<\alpha<1/2$ such that
\begin{equation}
\label{assumption_A4}
\tag{A4}
\begin{split}
& \sup_{i}\sup_{x^d \in F_d^{(i)}}\frac{\partial \log \pi^d(x^d)}{\partial x_i}=\bigO(d^{\alpha}),\quad
\sup_{x^d \in F_d^+}  \pi^d(x^d)=o(d^{1/2-\alpha}),\quad
\sup_{x^d \in F_d^+} 1/I_d(x^d)=\bigO(d^{\alpha/2}).
\end{split}
\end{equation}

\begin{remark}\label{remark_A4}
	For \ref{assumption_A4}, the first two conditions do not even require $\pi^d$ and the first partial derivative of $\log\pi^d$ to be bounded. Thus, they are quite weak. For the last condition, although the mode of $\pi^d$ is ruled out from $F_d^+$, the condition can hold as long as $\sup_{i}\sup_{x^d \in F_d^{(i)}}\frac{\partial \log \pi^d(x^d)}{\partial x_i}=\bigO(d^{\alpha/2})$ and $I_d(X^d)$ is tight. That is, $\forall  0<\epsilon<1$, there exists $K_{\epsilon}>0$ such that $\Pr(I_d(X^d)>K_{\epsilon})<1-\epsilon)$. To see this, one can choose $F_d$ using the tightness such that $\sup_{x^d\in F_d} 1/I_d(x^d)=\bigO(d^{\alpha/2})$. Then we can replace $F_d$ by $F_d^+$ since $\inf_{x^d\in F_d} I_d(x^d)-\inf_{x^d\in F_d^+} I_d(x^d)=\bigO(d^{\alpha/2}(\log d)^{1/2}d^{-1/2})=o(d^{-1/4})=o(d^{-\alpha/2})$. Note that $I_d(X^d)$ being tight is a very reasonable assumption, since if $I_d(X^d)$ is not tight, the target $\pi^d$ becomes ``flat'' at almost every state $x^d$. 
\end{remark}

We are now ready to present our first main result using the assumptions \ref{assumption_A1}, \ref{assumption_A2}, \ref{assumption_A3}, and  \ref{assumption_A4}. We establish the following results on asymptotic ESJD and asymptotic acceptance rate.

\begin{theorem} (Asymptotic ESJD and acceptance rate)\label{key_theorem}		
	Suppose $\pi^d$ satisfies \ref{assumption_A1}, \ref{assumption_A2}, \ref{assumption_A3}, and \ref{assumption_A4}, then as $d\to \infty$, we have
	\[
	&\left|\ESJD(d)-2\frac{d\ell^2}{d-1}\EE_{X^d\sim \pi^d}\left[ \Phi\left(-\frac{\ell \sqrt{I_d(X^d)}}{2}\right) \right]\right|\to 0,\\
	&\left|\EE_{X^d\sim\pi^d}\EE_{Y^d}\left(1\wedge \frac{\pi^d(Y^d)}{\pi^d(X^d)}\right)-2\EE_{X^d\sim \pi^d}\left[ \Phi\left(-\frac{\ell \sqrt{I_d(X^d)}}{2}\right) \right]\right|\to 0,
	\]
	where the expectation over $Y^d$ is taken for $Y^d\sim \mathcal{N}(x^d,\frac{\ell^2}{d-1}I)$ for given $X^d=x^d$.
\end{theorem}
\begin{proof}
	See \cref{proof_key_theorem}.
\end{proof}

Since the assumptions required by \cref{key_theorem} are very mild, the result of \cref{key_theorem} holds for a large class of realistic MCMC models. As an example, we give a class of graphical models that all conditions \ref{assumption_A1}, \ref{assumption_A2}, \ref{assumption_A3}, and \ref{assumption_A4} hold. Therefore, the asymptotic ESJD and acceptance rate by \cref{key_theorem} hold for this class of graphical models. We will further discuss realistic MCMC models later in \cref{subsection_verify} and  \cref{subsection_realistic}. 

{We give a simple criterion that the assumptions \ref{assumption_A1}, \ref{assumption_A2}, \ref{assumption_A3}, and \ref{assumption_A4} hold. More discussions and examples are delayed to \cref{section_examples}.}
	\begin{corollary} {If a graphical model satisfies (i) \emph{either} each node has at most  $l_d=o(d^{1/4})$ links \emph{or} the number of $3$-cliques of the graph is $o(d^{3/2})$; (ii)  $I_d(X^d)$ is tight; (iii) $\pi^d$ has bounded density and $\log \pi^d$ has up to the third bounded partial derivatives, then the assumptions \ref{assumption_A1}, \ref{assumption_A2}, \ref{assumption_A3}, and \ref{assumption_A4} hold. Therefore, the asymptotic ESJD and acceptance rate results by \cref{key_theorem} hold.}
	\end{corollary}	
\begin{proof}
{	First, the assumption \ref{assumption_A1} holds when second partial derivatives of $\log\pi^d$ are bounded. Next, the assumption \ref{assumption_A2} automatically holds for graphical models. Furthermore, $l_d=o(d^{1/4})$ implies that the number of $3$-cliques is $o(d^{3/2})$. Then one can easily verify that the assumption \ref{assumption_A3} holds using the fact that the third partial derivatives of $\log\pi^d$ are bounded. Finally, the assumption \ref{assumption_A4} holds since $I_d(X^d)$ is tight.}
\end{proof}

Note that \cref{key_theorem} suggests that under mild conditions on the target distribution, the expected acceptance rate 
\[
\EE_{X^d\sim\pi^d}\EE_{Y^d}\left(1\wedge \frac{\pi^d(Y^d)}{\pi^d(X^d)}\right)\to 2\EE_{X^d\sim \pi^d}\left[ \Phi\left(-\frac{\ell \sqrt{I_d(X^d)}}{2}\right) \right].
\]
Therefore, we can define asymptotic acceptance rate as a function of $\ell$ as follows.
\begin{definition}{(Asymptotic Acceptance Rate)}\label{def_asymp_accept_rate}
	The asymptotic acceptance rate function is defined by
	\[
	a(\ell):=2\EE_{X^d\sim \pi^d}\left[ \Phi\left(-\frac{\ell \sqrt{I_d(X^d)}}{2}\right) \right].
	\]
\end{definition}

The next theorem shows that if the target distribution satisfies \ref{assumption_A1}, \ref{assumption_A2}, \ref{assumption_A3} and \ref{assumption_A4}, then if we maximize the asymptotic ESJD, the resulting asymptotic acceptance rate is no larger than $0.234$. 
\begin{theorem}\label{thm_maximize_lower_bound}
	Defining  the optimal parameter for maximizing the asymptotic ESJD by $\hat{\ell}$, i.e.\
	\[
	\hat{\ell}:=\arg\max_{\ell} h(\ell), \quad h(\ell):= 2\ell^2\EE_{X^d\sim \pi^d}\left[\Phi\left(-\frac{\ell \sqrt{I_d(X^d)}}{2}\right)\right],
	\]
	then we have $a(\hat{\ell})\le 0.234$ {(to three decimal places).}
\end{theorem}
\begin{proof}
We follow the arguments in \cite[Lemma 5.1.4]{Tawn2017}. First, it can be verified by taking the second derivatives of $h(\ell)$ with respect to $\ell$ that the maximum of $h(\ell)$ is achieved at $\ell$ such that $\frac{\partial h(\ell)}{\partial \ell}=0$. Therefore, the optimal $\hat{\ell}$ satisfies
\[
2\EE_{X^d\sim \pi^d}\left[\Phi\left( -\frac{\hat{\ell}\sqrt{I_d(X^d)}}{2}\right)\right]=\EE_{X^d\sim \pi^d}\left[\frac{\hat{\ell}\sqrt{I_d(X^d)}}{2}\Phi'\left( -\frac{\hat{\ell}\sqrt{I_d(X^d)}}{2}\right)\right].
\]
Therefore, the asymptotic acceptance rate
\[
a(\hat{\ell})=\EE_{X^d\sim \pi^d}\left[\frac{\hat{\ell}\sqrt{I_d(X^d)}}{2}\Phi'\left( -\frac{\hat{\ell}\sqrt{I_d(X^d)}}{2}\right)\right]=\EE_{X^d\sim \pi^d}\left[-\Phi^{-1}(V)\Phi'\left(\Phi^{-1}(V)\right) \right],
\]
where $V:=\Phi\left(-\frac{\hat{\ell}\sqrt{I_d(X^d)}}{2}\right)$. By \cite{Sherlock2006}, the function $-\Phi^{-1}(x)\Phi'\left(\Phi^{-1}(x)\right)$
is a concave function for any $x\in (0,1)$. Therefore, we have
\[
a(\hat{\ell})=\EE_{X^d\sim \pi^d}\left[-\Phi^{-1}(V)\Phi'\left(\Phi^{-1}(V)\right) \right]\le -\Phi^{-1}[\EE_{X^d\sim \pi^d}(V)]\Phi'\left[\Phi^{-1}(\EE_{X^d\sim \pi^d}(V))\right].
\]
Defining $m:=-\Phi^{-1}[\EE_{X^d\sim \pi^d}(V)]$, we can then write
$a(\hat{\ell})=2\Phi(-m)\le m\Phi'(-m)$.
Finally, it suffices to show that $2\Phi(-m)\le m\Phi'(-m)$ implies $2\Phi(-m)\le 0.234$ (to three decimal places). Note that the function $x^2\Phi(-x)$ is maximized at $\hat{m}$ such that $2\Phi(-\hat{m})=\hat{m}\Phi'(-\hat{m})\approx 0.234$. By \cite[Lemma 5.1.4]{Tawn2017}, the function $2\Phi(-x)-x\Phi'(-x)$ is positive for $x<\hat{m}$ and negative for $x>\hat{m}$. Therefore, $2\Phi(-m)\le m\Phi'(-m)$ implies that $m>\hat{m}$. Since $\Phi(-x)$ is monotonically decreasing with $x$, we have
$a(\hat{\ell})=2\Phi(-m)\le 2\Phi(-\hat{m})\approx 0.234$.
\end{proof}

The next result is our main result for optimal scaling by maximizing ESJD. 
Defining the following WLLN condition for the target $\pi^d$:
\[\label{assumption_A5}\tag{A5}
{I_d(X^d)-\bar{I}_d\to 0}\quad \textrm{in probability}
\]
where $X^d\sim \pi^d$ and $\bar{I}_d:=\EE_{X^d\sim \pi^d}[I_d(X^d)]$,
we show that if the target distribution $\pi^d$ satisfies \ref{assumption_A1}, \ref{assumption_A2},  \ref{assumption_A3}, \ref{assumption_A4}, and the WLLN assumption in \ref{assumption_A5}, then the acceptance rate $0.234$ is asymptotically optimal.

\begin{theorem}(Optimal Scaling for Maximizing ESJD)\label{thm_ESJD}
	Suppose the target distribution $\pi^d$ satisfies \ref{assumption_A1}, \ref{assumption_A2}, \ref{assumption_A3}, \ref{assumption_A4}, and \ref{assumption_A5}. 
	Then
	the asymptotic optimal acceptance rate $a(\hat{\ell})\approx 0.234$ (to three decimal places).
\end{theorem}
\begin{proof} 
	By convexity of the function $\Phi(-x)$ when $x\ge 0$, we can immediately obtain a lower bound
	\[
	\ell^2\EE_{X^d\sim \pi^d}\left[ \Phi\left(-\frac{\ell \sqrt{I_d(X^d)}}{2}\right) \right]\ge \ell^2\left[ \Phi\left(-\frac{\ell\EE_{X^d\sim \pi^d}[ \sqrt{I_d(X^d)}]}{2}\right) \right].
	\]
	Under \ref{assumption_A5}, this lower bound is asymptotically tight.
	Therefore, as $d\to \infty$, according to \cite{Roberts1997}, we have (to two decimal places)
	\[
	\hat{\ell}\to \frac{2.38}{\EE_{X^d\sim \pi^d}[ \sqrt{I_d(X^d)}]},\quad h(\hat{\ell})\to \frac{1.3}{\left(\EE_{X^d\sim \pi^d}[ \sqrt{I_d(X^d)}]\right)^2}.
	\]
	The acceptance rate which maximizing the asymptotic ESJD is
	\[
	a(\hat{\ell})&=2\EE_{X^d\sim \pi^d}\left[ \Phi\left(-\frac{\hat{\ell} \sqrt{I_d(X^d)}}{2}\right) \right]\to  2\Phi\left(-\frac{\hat{\ell} \EE_{X^d\sim \pi^d}\sqrt{I_d(X^d)}}{2}\right) \\
	&\approx 2\Phi\left(-\frac{2.38}{\EE_{X^d\sim \pi^d}[ \sqrt{I_d(X^d)}]}\frac{ \EE_{X^d\sim \pi^d}[\sqrt{I_d(X^d)}]}{2}\right)= 2\Phi(-1.19)\approx 0.234.
	\]
\end{proof}

\begin{remark}\label{remark_A5}
	Comparing the results of \cref{thm_maximize_lower_bound} and \cref{thm_ESJD}, it is clear that the ``roughness'' of $\pi^d$, $I_d(X^d)$, is the key quantity which determines the optimal acceptance rate $a(\hat{\ell})\le 0.234$ when only the tightness of $I_d(X^d)$ can be verified, or $a(\hat{\ell})\approx 0.234$ when the concentration of $I_d(X^d)$ as defined in \ref{assumption_A5} can be verified. We will later demonstrate how to verify \ref{assumption_A5} for some realistic MCMC models in \cref{subsection_verify} and  \cref{subsection_realistic}. 
\end{remark}

\subsection{Optimal Scaling via Diffusion Limits}\label{subsection_diffusion}
In this subsection, we consider sufficient conditions on $\pi^d$ for establishing weak convergence of diffusion limits. As we discussed before, establishing such results gives stronger guarantee for optimal scaling than maximizing ESJD. However, it also requires stronger conditions on the target distribution. As we will see in the following, we need to strengthen assumptions \ref{assumption_A2}, \ref{assumption_A3}, \ref{assumption_A4}, \ref{assumption_A5} and add one more assumption \ref{assumption_A6}.

We first strengthen \ref{assumption_A2} to a new assumption \ref{assumption_A2_plus} as follows.
\[\label{assumption_A2_plus}
\tag{A2+}
\sum_{i=1}^d\sum_{j=1}^d\sum_{k=1}^d\int \left(\frac{\partial^2 \pi^d}{\partial x_i^2}\frac{1}{\pi^d}\right)\left(\frac{\partial^2 \pi^d}{\partial x_j^2}\frac{1}{\pi^d}\right)\left(\frac{\partial^2 \pi^d}{\partial x_k^2}\frac{1}{\pi^d}\right)\pi^d \dee x^d=\bigO(d^{2-\delta})
\]
for some $\delta>0$.
\begin{remark}
The new assumption \ref{assumption_A2_plus} is stronger than \ref{assumption_A2} but is still very mild. To see this, we consider graphical models as examples. For graphical models with $d$ nodes each with $\bigO(l_d)$ links, there are at most $\bigO(dl_d^2)$ $3$-cliques. Therefore, \ref{assumption_A2_plus} holds for any graphical model with $l_d=o(d^{1/2-\delta})$ and bounded second partial derivatives of $\log\pi^d$. Note that this is only for the worst case, as many realistic graphical models do not have $3$-cliques.
\end{remark}
Next, we slightly strengthen \ref{assumption_A3} and \ref{assumption_A4} to \ref{assumption_A3_plus} and \ref{assumption_A4_plus}.  
\begin{equation}
\label{assumption_A3_plus}
\tag{A3+}
\begin{split}
&\sup_{(i,j): j\notin \mathcal{H}_i} \sup_{x^d \in {\Reals^d}}\frac{\partial^3 \log \pi^d(x^d)}{\partial x_i^2 \partial x_j}=o(1),\quad
\sup_{(i,j): j\in\mathcal{H}_i} \sup_{x^d \in {\Reals^d}}\frac{\partial^3 \log \pi^d(x^d)}{\partial x_i^2 \partial x_j}=o(\sqrt{d/l_d}),\\	
&{\sum_{i\neq j\neq k}\left(\sup_{x^d\in {\Reals^d}}\left|\frac{\partial^3 \log \pi^d(x^d)}{\partial x_i\partial x_j \partial x_k}\right|\right)=o(d^{3/2}).}
\end{split}
\end{equation}
Suppose exists $0<\alpha<1/2$ that
\begin{equation}
\label{assumption_A4_plus}
\tag{A4+}
\begin{split}
&\sup_{i}\sup_{x^d \in F_d^{(i)}}\frac{\partial^2 \log \pi^d(x^d)}{\partial x_i^2}=o(d^{\alpha}),\quad
\sup_{i}\sup_{x^d \in F_d^{(i)}}\frac{\partial \log \pi^d(x^d)}{\partial x_i}=\bigO(d^{\alpha/2}),\\
&\sup_{x^d \in F_d^+}  \pi^d(x^d)=o(d^{1/2-\alpha}),\quad
\sup_{x^d \in F_d^+} 1/I_d(x^d)=\bigO(d^{\alpha/4}).
\end{split}
\end{equation}
Furthermore, we strengthen the WLLN condition \ref{assumption_A5} to the following \ref{assumption_A5_plus}.
\[\label{assumption_A5_plus}
\tag{A5+}
\sup_{x^d \in F_d^+} \left|I_d(x^d)-\bar{I}\right|\to 0
\]
where $\bar{I}:=\lim_{d\to \infty} \bar{I}_d$ exists.

\begin{remark}
	\ref{assumption_A3_plus} is only slightly stronger than \ref{assumption_A3} on the rates. \ref{assumption_A4_plus} also includes a new condition on the rate of $\frac{\partial^2 \log \pi^d(x^d)}{\partial x_i^2}$ which is quite weak. \ref{assumption_A5_plus} requires any sequence $(x^1,x^2,\dots,x^d,\dots)$ where $x^i\in F_i^+$ converges to the same limit $\bar{I}$, so it is (slightly) stronger than WLLN condition in \ref{assumption_A5}. It will become clear in the proof of \cref{thm_diffusion} that \ref{assumption_A5_plus} is to ensure the speed measure of the diffusion process $h(\ell)$ does not depend on the state $x^d$.
\end{remark}	
Finally, we define a new assumption \ref{assumption_A6} on the target distribution. Roughly speaking, the new assumption is to require the first coordinate of $\pi^d$ is asymptotically independent with the rest.
\begin{equation}\label{assumption_A6}\tag{A6}
\begin{split}
&{\lim_{d\to \infty}}\sup_{x^d\in F_d^{+}}\left|{\frac{\dee }{\dee x_1}\left[\log \pi^d(x_1\,|\,x_{-1})-\log {\tilde{\pi}}(x_1)\right]}\right|{=} 0,
\end{split}
\end{equation}
where $x_{-1}:=(x_2,\dots,x_d)$, {$\tilde{\pi}$ is a one-dimensional density and $(\log {\tilde{\pi}})'$ is Lipschitz continous.}

\begin{remark}
	Note that  \ref{assumption_A6} is a strong condition, which may not be satisfied for many realistic MCMC models. However, it might be necessary in order to get a one-dimensional diffusion limit for the first coordinate. In the proof of the optimal scaling via diffusion limits result in \cref{thm_diffusion}, the assumption \ref{assumption_A6} is to ensure the SDE for the first coordinate $x_1$ doesn't depend on the values of other coordinates. Furthermore, although we do not pursue in this paper, if in \ref{assumption_A6} we instead assume not just the first component but 
	a finite collection of components are asymptotically independent from the rest, a version of weak convergence to multi-dimensional diffusion limits could be obtained following similar arguments as the proof of the one-dimensional diffusion limit case in \cref{thm_diffusion}. 
\end{remark}
Now we are ready for the main result of optimal scaling via diffusion limits, which is given in \cref{thm_diffusion}. 
{We show that, even though different coordinates of the Markov chain are \emph{not} independent \emph{nor} even individually Markovian, the sped-up first-coordinate process converges to a limiting diffusion limit under much more general conditions on the target distribution.}
Comparing with the assumptions in \cref{thm_ESJD}, the new sufficient conditions for diffusion limits include strengthening \ref{assumption_A2} to \ref{assumption_A2_plus}, \ref{assumption_A3} and \ref{assumption_A4} to \ref{assumption_A3_plus} and \ref{assumption_A4_plus}, \ref{assumption_A5} to \ref{assumption_A5_plus}, and adding \ref{assumption_A6}. We also require slightly stronger condition on the sequence of ``sets of typical states'' $\{F_d\}$.
\begin{theorem}{(Optimal Scaling via Diffusion Limits)}\label{thm_diffusion}
	Suppose the sequence $\{F_d\}$ satisfies
	{$\pi^d(F_d^c)=\bigO(d^{-1-\delta})$}
	for some $\delta>0$, the target distribution $\pi^d$ satisfies \ref{assumption_A1}, \ref{assumption_A2_plus}, \ref{assumption_A3_plus}, \ref{assumption_A4_plus}, \ref{assumption_A5_plus}, and \ref{assumption_A6},
	then for $U^d(t):=X_1^d(\lfloor d t\rfloor)$, as $d\to\infty$, we have
	$U^d\Rightarrow U$,
	where $\Rightarrow$ denotes weak convergence in Skorokhod topology,  and $U$ satisfies the Langevin SDE
	\[\label{equ_SDE}
	\dee U(t)=(h(\ell))^{1/2}\dee B(t)+h(\ell){\frac{\tilde{\pi}'(U(t))}{2\tilde{\pi}(U(t))}}\dee t,
	\]
	where $h(\ell):=2\ell^2\Phi(-\ell\sqrt{\bar{I}}/2)$ is the speed measure for the diffusion process.
\end{theorem}
\begin{proof}
	See \cref{proof_thm_diffusion}.
\end{proof}
\begin{remark}
	Note that \cref{thm_diffusion} allows dependent coordinates on the target distribution, which is much more general than the product i.i.d.\ condition. The only strong assumption is  \ref{assumption_A6} which requires the first coordinate is asymptotically independent with other coordinates.
\end{remark} 


Next, we present another result with slightly stronger conditions, which allows the RWM algorithm to start at a fixed state. This stronger convergence result later allows us to establish a complexity bound for the RMW algorithm in \cref{subsection_complexity_bounds}
{Let $X^d=\{X^d(t), t\ge 0\}$ for $d\in \mathbb{N}$ be the RWM processes defined earlier. Without loss of generality, suppose $\{X^d, d=1,2,\dots\}$ are defined in a common measurable metric space $(\Reals^{\infty},\mathcal{F},\rho)$ as independent processes. }

\begin{theorem}{(Optimal Scaling via Diffusion Limits with fixed starting state)}\label{thm_diffusion2}
	Suppose $X_1^d$ converges weakly in the Skorokhod topology as $d\to\infty$ to a c\`adl\`ag process $X_1^{\infty}$. Moreover, assume these processes $\{X^d, d=1,2,\dots\}$ all have the same marginal stationary distribution $\pi_1$ for the first coordinate and that the first coordinate of $X^{\infty}$ converges {weakly}
	to $\pi_1$. 
	Suppose the sequence $\{F_d\}$ satisfies
	{$\pi^d(F_d^c)=\bigO(d^{-2-\delta})$}
	for some $\delta>0$, the target distribution $\pi^d$ satisfies \ref{assumption_A1}, \ref{assumption_A3_plus}, \ref{assumption_A4_plus}, \ref{assumption_A5_plus}, and \ref{assumption_A6}. We strengthen \ref{assumption_A2_plus} to the following condition
	\begin{equation}\label{assumption_A2_pp}\tag{A2++}
	\sum_{i,j,k,l,m\in \{2,\dots,d\}}\int \left(\frac{\partial^2 \pi_{-1}}{\partial x_i^2}\cdot\frac{\partial^2 \pi_{-1}}{\partial x_j^2}\cdot\frac{\partial^2 \pi_{-1}}{\partial x_k^2}\cdot\frac{\partial^2 \pi_{-1}}{\partial x_l^2}\cdot\frac{\partial^2 \pi_{-1}}{\partial x_m^2}\right)\left(\frac{1}{\pi_{-1}}\right)^5\pi^d \dee x^d=\bigO(d^{3-6\delta}).
	\end{equation}
	Then as $d\to\infty$, we have
		${}_xU^d\Rightarrow {}_xU$,
	where ${}_xU^d(t):=(X_1^d(\lfloor d t\rfloor)\,|\, X_1^d(0)=x)$ is the first coordinate of the RWM algorithm sped up by a factor of $d$, conditional on starting at the state $x$, and ${}_xU$ is the limiting ergodic Langevin diffusion $U$ in \cref{equ_SDE} also conditional on starting at $x$.
\end{theorem}
\begin{proof}
	See \cref{proof_thm_diffusion2}.
\end{proof}

\begin{remark}
	The new assumption \ref{assumption_A2_pp} is stronger than \ref{assumption_A2_plus} but is still not strong. To see this, for graphical models with $d$ nodes, each with $\bigO(l_d)$ links, we have at most $\bigO(dl_d^2)$ $3$-cliques. Under flat tail assumptions, at most $\bigO(d^2l_d^3)$ terms in the summation in \ref{assumption_A2_pp} is not zero. Therefore, \ref{assumption_A2_pp} holds for any graphical model with $l_d=o(d^{1/3-2\delta})$ and bounded second partial derivatives of $\log\pi^d$. Note that this is only for the worst case, as many realistic graphical models do not have $3$-cliques.
\end{remark}

\subsection{Complexity Bounds via Diffusion Limits}\label{subsection_complexity_bounds}


In the following, by
combing \cref{thm_diffusion2} and \cref{lemma_RR16}, we present a complexity bound for the RWM algorithm which holds for much more general target distributions comparing with \cite{Roberts2016}. More specifically, if the target distribution satisfies the conditions given in \cref{thm_diffusion2} which allows dependent coordinates of the target distribution, the RWM algorithm in $d$ dimensions takes $\bigO(d)$ iterations to converge to stationarity.
\begin{corollary}{(Complexity Bound for RWM Algorithms)}\label{corollary_complexity}
	{Under the conditions of \cref{thm_diffusion2}, for any $\epsilon>0$, there exists $D<\infty$ and $T<\infty$, such that
	\[
	\EE_{X_1^d(0)\sim \pi_1}\|\mathcal{L}_{X_1^d(0)}(X^d_1(\lfloor dt\rfloor))-\pi_1\|_{\textrm{KR}}<\epsilon,\quad \forall t\ge T, d\ge D,
	\]
	where $\pi_1$ denotes the marginal stationary distribution of the first coordinate.}
\end{corollary}
\begin{proof}
	The result directly comes from \cref{lemma_RR16} and \cref{thm_diffusion2}.
\end{proof}

\section{Examples and Applications}\label{section_examples}

In this section, we further discuss examples and applications of the main results in \cref{section_main_results}. We first discuss in \cref{subsection_verify} on verifying the assumptions of \cref{thm_ESJD} for realistic MCMC models. We have explained in \cref{remark_A1,remark_A2,remark_A3,remark_A4} that  \ref{assumption_A1}, \ref{assumption_A2}, \ref{assumption_A3}, and \ref{assumption_A4} are typically very weak conditions and they hold for some classes of graphical models. However, as discussed in \cref{remark_A5}, the assumption  \ref{assumption_A5} may need to be verified case by case. Particularly, in order to satisfy \ref{assumption_A5}, we may need to make additional assumptions on the observed data. Fortunately, we show by a simple Gaussian example in \cref{example_toy} that, in some cases, \ref{assumption_A5} can be easily verified without any further assumptions. 
Then, in \cref{subsection_realistic},
{we extend the simple Gaussian example in \cref{example_toy} to a more realistic MCMC model in \cref{example_realistic} and show it satisfies all the assumptions required by \cref{thm_ESJD}.} Thus, the acceptance rate $0.234$ is indeed asymptotically optimal for this realistic MCMC model.

\subsection{Discussions on \cref{thm_ESJD}}\label{subsection_verify}

The optimal scaling result for maximizing ESJD in \cref{thm_ESJD} requires one to verify that the target distribution satisfies \ref{assumption_A1}, \ref{assumption_A2}, \ref{assumption_A3}, \ref{assumption_A4}, and \ref{assumption_A5}. We discuss how to verify the conditions on the target distribution required by \cref{thm_ESJD} in practice. We explain that \ref{assumption_A1}, \ref{assumption_A2}, \ref{assumption_A3} and \ref{assumption_A4} are quite mild and usually easy to be verified. Therefore, we usually only need to focus on the
WLLN condition in  \ref{assumption_A5}, which might be difficult to check in practice. Throughout this subsection, we demonstrate verification of all the assumptions by a simple Gaussian example, which can be seen as a simplified version of typical Bayesian hierarchical models.
\begin{example}{(A Gaussian Example)}\label{example_toy}
	Consider a simple Gaussian MCMC model
	\[
	\begin{split}
	Y_{ij}~|~\theta_{ij}\quad &\sim \mathcal{N}(\theta_{ij},1),\quad i,j\in\{1,\dots, n\}\\
	\theta_{ij}~|~\mu_j &\sim \mathcal{N}(\mu_j,1), \quad i\in \{1,\dots, n\}\\
	\mu_j~|~\nu &\sim \mathcal{N}(\nu,1)\\
	\nu &\sim \textrm{ flat prior on }\mathbb{R},\\
	\end{split}
	\]
	where $\{Y_{ij}\}_{i,j=1}^n$ are the observed data, and $x^d=(\nu,\{\mu_j\}_{j=1}^n, \{\theta_{ij}\}_{i,j=1}^n)$ are parameters. Note that we have the number of parameters $d=n^2+n+1$ in this example. The target distribution (i.e. the posterior distribution) satisfies
	\[
	\begin{split}
	\pi^d(x^d)&=\Pr(x^d~|~ \{Y_{ij}\}_{i,j=1}^n)
	\propto \prod_{j=1}^n\prod_{i=1}^n\frac{1}{\sqrt{2\pi }}e^{-\frac{(\mu_j-\nu)^2}{2}}\frac{1}{\sqrt{2\pi }}e^{-\frac{(\theta_{ij}-\mu_j)^2}{2}}\frac{1}{\sqrt{2\pi}}e^{-\frac{(Y_{ij}-\theta_{ij})^2}{2}}.
	\end{split}
	\]
Note that the hyperparameters $\nu$ is conditionally independent given $\{\theta_{ij}\}$. Therefore, $\nu$ is only directly dependent with $n$ coordinates $\{\mu_j\}_{j=1}^n$. We can define the ``neighborhoods'' of $\nu$ using the collection of $\mu_j, j=1,\dots,n$.  Similarly, $\mu_j$ is directly dependent with $\nu$ and $\{\theta_{ij}\}_{i=1}^n$ and $\theta_{ij}$ is directly dependent with $\mu_j$. Therefore, if we choose the directly dependent coordinates as ``neighborhoods'', we have $l_d=n+1=\bigO(d^{1/2})$.	
\end{example}	

\subsubsection{Verifying  \ref{assumption_A1} to \ref{assumption_A4}}
First of all, the two conditions for $(i,j): j\neq \mathcal{H}_i$ in \ref{assumption_A1} and \ref{assumption_A3} hold trivially for graphical models.
Furthermore, in \cref{example_toy}, the parameter $\nu$ is conditional independent with all $\theta_{ij}$ and the corresponding conditional posterior distributions all have Gaussian tails, which implies \ref{assumption_A2} holds for any pair of coordinates $(\nu,\theta_{ij})$. Similarly, one can easily verify the assumption holds for other pairs of parameters.

Next, all the conditions on the third partial derivatives of $\log \pi^d$ hold, since there is no $3$-cliques. Moreover, in \cref{example_toy}, we have $l_d=\bigO(d^{1/2})$. The second partial derivative is $\bigO(1)$, and the density $\pi^d$ is bounded, so 
 the following conditions hold without the need of choosing $\{F_d\}$:
\[
&\sup_{(i,j): j\in \mathcal{H}_i}\sup_{x^d\in F_d^+} \frac{\partial^2 \log \pi^d(x^d)}{\partial x_i \partial x_j}=o(\sqrt{d/l_d}),\quad
\sup_{x^d \in F_d^+}  \pi^d(x^d)=o(d^{1/2-\alpha}).
\]

Finally, the last two conditions are almost immediately true once \ref{assumption_A5} has been verified: 
\[
& \sup_{i\in\{1,\dots,d\}}\sup_{x^d \in F_d^+}\frac{\partial \log \pi^d(x^d)}{\partial x_i}=\bigO(d^{\alpha}),\quad
\sup_{x^d \in F_d^+}1/I(x^d)=\bigO(d^{\alpha/2}).
\]
To see this, under \ref{assumption_A5}, we have
$\frac{1}{d}\sum_{i=1}^d \left(\frac{\partial}{\partial x_i}\log \pi^d(x^d) \right)^2 \to \bar{I}_d$. If $\bar{I}_d\to \bar{I}$ and $\bar{I}>0$, then we can select constant $K_2>0$ small enough such that $\bar{I}>K_2 d^{-\alpha/2}>0$ then $\bar{I}_d> K_2d^{-\alpha/2}$ for all large enough $d$. Next, by choosing the typical set $F_d$ such that for any $x^d\in F_d^+$, we have
$\frac{\partial \log \pi^d(x^d)}{\partial x_i} \le K_1 d^{\alpha},\quad I_d(x^d)\ge K_2d^{-\alpha/2}$,
where $K_1$ is a large enough constant. Then it suffices to check if $\{F_d\}$ is a valid sequence of typical sets such that
{$\pi^d(F_d)\to 1$.}
 For \cref{example_toy}, we have $X^d=(\nu,\{\nu_j\}_{j=1}^n,\{\theta_{ij}\}_{i,j=1}^n)$. We will show later that \ref{assumption_A5} holds such that under $X^d\sim \pi^d$ we have
$\frac{1}{d}\sum_{i=1}^d \left(\frac{\partial}{\partial x_i}\log \pi^d(X^d) \right)^2 \to 3$.
For example, we can choose $K_2=0.01$, $K_1=100$, and the typical set $F_d$ such that, for any $X^d=x^d\in F_d^+$, we have
\[
&I_d(x^d)>0.01 n^{-\alpha},\quad \frac{\partial \log \pi^d}{\partial \nu }=n(\bar{\mu}-\nu)\le 100n^{2\alpha},\\
& \frac{\partial \log \pi^d}{\partial \mu_j}=(n+1)\left(\frac{\sum_i \theta_{ij}+\nu}{n+1}-\mu_j\right)\le 100n^{2\alpha},\\
&\frac{\partial \log \pi^d}{\partial \theta_{ij}}=2\left(\frac{Y_{ij}+\mu_j}{2}-\theta_{ij}\right)\le 100 n^{2\alpha},
\] 
where $\alpha<1/2$ can be arbitrarily close to $1/2$. 	Observing that, under $X^d\sim \pi^d$, we have the following conditional distributions.
\begin{equation}\label{eq_observe_cond_iid}
\begin{split}
\theta_{ij}~|~Y_{ij},\mu_j\quad &\sim^{\textrm{indep.}} \mathcal{N}\left(\frac{\mu_j+Y_{ij}}{2},\frac{1}{2}\right),\quad i,j\in\{1,\dots, n\},\\
\mu_j~|~\sum_i \theta_{ij}, \nu &\sim^{\textrm{indep.}} \mathcal{N}\left(\frac{\sum_i \theta_{ij}+\nu}{n+1},\frac{1}{n+1}\right), \quad i\in \{1,\dots, n\},\\
\nu~|~\bar{\mu} &\sim \mathcal{N}\left(\bar{\mu},\frac{1}{n}\right).
\end{split}
\end{equation}
Then it can be easily verified that {$\pi^d(F_d)\to 1$.} 
\subsubsection{Verifying  \ref{assumption_A5}}

One assumption of \cref{thm_ESJD} that could be difficult to verify in practice is \ref{assumption_A5}.
It requires the sequence of random variables $\{I_d(X^d)\}$ converge to a sequence of constants in probability. We feel this assumption has to be checked case by case and it is hard to get general sufficient condition for it to hold. For realistic MCMC models, this may require assumptions on the observed data so that the posterior distribution has certain ``concentration'' properties as $d\to \infty$.  

Fortunately, for \cref{example_toy}, we can verify that  \ref{assumption_A5} holds without any further assumption on the observed data $\{Y_{ij}\}$. Note that in \cref{example_toy}, we have
	\[
	&\left(\frac{\partial \log \pi^d}{\partial \nu}\right)^2=\left(\sum_{j}(\mu_j-\nu)\right)^2=n^2\left( \bar{\mu}-\nu\right)^2,\\
	&\left(\frac{\partial \log \pi^d}{\partial \mu_j}\right)^2=\left(\sum_i (\theta_{ij}-\mu_j) -(\mu_j-\nu)\right)^2=(n+1)^2\left(\frac{\sum_i \theta_{ij}+\nu}{n+1}-\mu_j\right)^2,\\
	&\left(\frac{\partial \log \pi^d}{\partial \theta_{ij}}\right)^2=\left((Y_{ij}-\theta_{ij})-(\theta_{ij}-\mu_j)\right)^2=4\left(\frac{Y_{ij}+\mu_j}{2}-\theta_{ij}\right)^2.
	\]
	Hence, if suffices to show that, under $X^d=(\nu,\{\mu_j\}_{j=1}^n,\{\theta_{ij}\}_{i,j=1}^n)\sim \pi^d$, the following three terms converges to some constants in probability or in distribution:
	\[
	&\frac{1}{d}\left(\frac{\partial \log \pi^d}{\partial \nu}\right)^2=\frac{n^2}{n^2+n+1}(\bar{\mu}-\nu)^2,\\
	&\frac{1}{d}\sum_j\left(\frac{\partial \log \pi^d}{\partial \mu_j}\right)^2=\frac{(n+1)^2}{n^2+n+1}\sum_j\left(\frac{\sum_i \theta_{ij}+\nu}{n+1}-\mu_j\right)^2,\\
	&\frac{1}{d}\sum_{ij}\left(\frac{\partial \log \pi^d}{\partial \theta_{ij}}\right)^2=\frac{4}{d}\sum_{ij}\left(\frac{Y_{ij}+\mu_j}{2}-\theta_{ij}\right)^2.
	\]
	We have observed that the target distribution $\pi^d$ has conditional independence structure in \cref{eq_observe_cond_iid},
	which immediately  leads to
	\[
	&(\bar{\mu}-\nu)^2\to^{\Pr} 0,\quad
	\sum_j\left(\frac{\sum_i\theta_{ij}+\nu}{n+1}-\mu_j\right)^2\to^{\Pr} 1,\quad
	\frac{1}{d}\sum_{ij}\left(\frac{Y_{ij}+\mu_j}{2}-\theta_{ij}\right)^2\to^{\Pr} \frac{1}{2}.
	\]
	Therefore, \ref{assumption_A5} is satisfied.
	
	Overall, we have checked all the assumptions of \cref{thm_ESJD} for our simple Gaussian example. Therefore, by \cref{thm_ESJD}, we have the following optimal scaling result for \cref{example_toy}.
	\begin{proposition}
		The optimal scaling for \cref{example_toy} by maximizing ESJD is to choose (to two decimal places)
		$\hat{\ell}\approx\frac{2.38}{\EE_{X^d\sim \pi^d}[ \sqrt{I(X^d)}]}\to \frac{2.38}{\sqrt{3}}\approx 1.37$
		and the corresponding asymptotic acceptance rate is (to three decimal places) $0.234$.
	\end{proposition}


\subsection{Optimal Scaling of a Realistic MCMC Model}\label{subsection_realistic}
We first
discuss sufficient conditions for two more classes of graphical models. In \cref{corollary_jeffrey}, we give sufficient conditions for the first equation of \ref{assumption_A1}, \ref{assumption_A2}, and the first equation of \ref{assumption_A3} to hold for one particular class of graphical models. In \cref{corollary_aaron}, we give sufficient conditions for \ref{assumption_A5} to hold for one specific class of graphical models.

First, we consider the class of graphical models represented by the factor graphs:
\[\label{graphical_model0}
\pi^d(x^d)\propto\prod_{k=1}^{K_d}\psi_k(\{x_i: i\in C_k\}),
\]
where $C_k$ are cliques, $\psi_k$ are potentials, $K_d$ denotes the number of potentials.
\begin{proposition}\label{corollary_jeffrey}
	For the class of graphical models represented by \cref{graphical_model0}. Let $m_d$ denotes  the  maximum number of cliques a coordinate can belong to. If all the potentials $\psi_k$ have ``flat tails'' in the sense that for all $k$ we have $\frac{\partial \psi_k}{\partial x_i}\to 0$ as $|x_i|\to\infty$ for all  $i\in C_k$, and the cardinality of $C_k$ satisfies $\sup_k |C_k| =o(d/m_d)$, then the first equation in \ref{assumption_A1}, \ref{assumption_A2}, and the first equation in \ref{assumption_A3} hold.
\end{proposition}

Next, we consider Bayesian hierarchical modeling where $K$ denotes the number of ``layers'' or ``stages'' of the model. We use $\theta^{(k)}, k=1,\dots,K$ to denote the parameter vector with length $n_k$ for the $k$-th layer, where $\theta^{(k)}:=(\theta^{(k)}_1,\dots,\theta^{(k)}_{n_k})$. We consider the special structure of the graphical model such that $\theta^{(k)}$ is only connected to $\theta^{(k-1)}$ and $\theta^{(k+1)}$. Using factor graphs, let $x^d=(\theta^{(1)},\dots,\theta^{(K)})$ we can represent the target distribution as
\[\label{graphical_model}
\pi^d(x^d)\propto\prod_{k=1}^K \psi_k(\theta^{(k-1)},\theta^{(k)}),
\]
where $d=\sum_{k=1}^K n_k$, $\{\psi_k\}$ are the potentials, and without loss of generality we assumed $\theta^{(0)}$ to be the observed data. 

In the following, we show that  \ref{assumption_A5} hold for the class of graphical models represented by \cref{graphical_model} under certain conditions.
\begin{proposition}\label{corollary_aaron}
	For the class of graphical models represented by \cref{graphical_model}, if
	$\theta^{(k)}=(\theta^{(k)}_1,\theta^{(k)}_2,\dots,\theta^{(k)}_{n_k})$ are independent conditional on $\theta^{(k-1)}$ and $\theta^{(k+1)}$ and this holds for all $k$. Moreover, if under $X^d=(\theta^{(1)},\dots,\theta^{(K)})\sim\pi^d$ all the potentials $\psi_k$ satisfy
	\[
	\sup_{i \in \{1,\dots,n_k\}}\left|\frac{\partial \log \psi_k}{\partial \theta^{(k)}_i}\right|=\bigO_{\Pr}\left(\sqrt{d/n_k}\right), \quad \sup_{j \in \{1,\dots,n_{k-1}\}}\left|\frac{\partial \log \psi_k}{\partial \theta^{(k-1)}_j}\right|=\bigO_{\Pr}\left(\sqrt{d/n_{k-1}}\right)
	\] 
	then  \ref{assumption_A5} holds.
\end{proposition}

Next, we extend the simple Gaussian example in \cref{example_toy} to a more realistic MCMC model which belongs to both classes of graphical models in \cref{graphical_model0,graphical_model} and show that all the assumptions for the optimal scaling result in \cref{thm_ESJD} hold.
\begin{example}{(A Realistic MCMC Model)}\label{example_realistic}
	Consider a realistic MCMC model
	\[
	\begin{split}
	Y_{ij}~|~\theta_{ij}\quad &\sim \mathcal{N}(\theta_{ij},W),\quad i,j\in\{1,\dots, n\}\\
	\theta_{ij}~|~\mu_j &\sim \mathcal{N}(\mu_j,V), \quad i\in \{1,\dots, n\}\\
	\mu_j~|~\nu &\sim \mathcal{N}(\nu,A)\\
	\nu &\sim \textrm{ flat prior on }\mathbb{R},\\
	A &\sim \textbf{IG}(a,b),
	\end{split}
	\]
	where $x^d=(\nu,A,\{\mu_j\}_{j=1}^n, \{\theta_{ij}\}_{i,j=1}^n)$ are parameters, $\{Y_{ij}\}$ are the observed data, and $a,b,W,V$ are known constants. 
	\end{example}

	We further assume that  the observed data $\{Y_{ij}\}$ is not abnormal so that the  posterior of the hyperparameter $A$ concentrates to some unknown constant.
	\begin{assumption*}{}
			The  posterior of the hyperparameter $A$ in \cref{example_realistic} concentrates to some unknown constant $A_0>0$ as $n\to\infty$.
	\end{assumption*}
	Note that this is a very reasonable assumption which implies the MCMC model is not seriously misspecified. We do not discuss sufficient conditions on the observed data $\{Y_{ij}\}_{i,j=1}^n$ for concentration of posterior distribution of $A$ here since it is not the focus of this paper. Next, we show that, under this assumption, the realistic MCMC model satisfies all the conditions required for optimal scaling in \cref{thm_ESJD}. Therefore, the acceptance rate $0.234$ is indeed asymptotically optimal for this MCMC model in the sense of maximizing ESJD.

\begin{proposition}\label{prop_realistic_example}
	Under the above assumption, the optimal asymptotic acceptance rate for the realistic MCMC model in \cref{example_realistic} is (to three decimal places) $0.234$.
\end{proposition}
\begin{proof}
	See \cref{proof_prop_realistic_example}.		
\end{proof}


\section*{Acknowledgment}
The authors thank Jeffrey Negrea for suggestions on graphical models, and Aaron Smith for helpful discussions. J.~R.\ is partly supported by the Natural Sciences and Engineering Research Council (NSERC) of Canada.

\printbibliography

\appendix

\renewcommand{\refname}{Appendix}

\section{Proof of \cref{key_theorem}}\label{proof_key_theorem}

Throughout the proof, for simplicity, we assume the coordinates are linear ordered. The ``neighborhoods'' of a coordinate is defined by  $\mathcal{H}_i:=\{j: |i-j|<l_d\}$. Therefore $\sup_{(i,j): j\in \mathcal{H}_i}$ can be simplified to $\sup_{|i-j|<l_d}$ and $\sup_{(i,j): j\notin \mathcal{H}_i}$ can be simplified to $\sup_{|i-j|\ge l_d}$. Note that the use of linear ordering is only for simplifying notations. It is straightforward to extend the proof to the cases of general ordering.

For \cref{key_theorem}, we only prove
\[
\left|\ESJD(d)-2\frac{d\ell^2}{d-1}\EE_{X^d\sim \pi^d}\left[ \Phi\left(-\frac{\ell \sqrt{I_d(X^d)}}{2}\right) \right]\right|\to 0,
\]
since the proof of 
\[
\left|\EE_{X^d\sim\pi^d}\EE_{Y^d}\left(1\wedge \frac{\pi^d(Y^d)}{\pi^d(X^d)}\right)-2\EE_{X^d\sim \pi^d}\left[ \Phi\left(-\frac{\ell \sqrt{I_d(X^d)}}{2}\right) \right]\right|\to 0
\]
follows similarly.

First, we write $\ESJD$ as
$\ESJD(d)=:\sum_{i=1}^d\ESJD_i(d)$,
where 
\[
\ESJD_i(d):=\EE_{X^d\sim\pi^d}\EE_{Y^d}\left[(Y_i-X_i)^2\left(1\wedge \frac{\pi^d(Y^d)}{\pi^d(X^d)}\right)\right].
\]
Then it suffices to show that 
\[
\sup_{i\in\{1,\dots,d\}}\left|\ESJD_i(d)-\frac{2\ell^2}{d-1}\EE_{X^d\sim \pi^d}\left[ \Phi\left(-\frac{\ell \sqrt{I_d(X^d)}}{2}\right) \right]\right|=o(d^{-1}).
\]
Writing 
$\ESJD_i(d)=\EE_{X^d\sim\pi^d}\EE_{Y_i}\left[(Y_i-X_i)^2\EE_{Y_{-i}}\left(1\wedge \frac{\pi^d(Y^d)}{\pi^d(X^d)}\right)\right]$, it suffices to show that uniformly over $i\in\{1,\dots,d\}$
\[
&\EE_{X^d\sim \pi^d}\left|\EE_{Y_i}\left[(Y_i-X_i)^2\EE_{Y_{-i}}\left(1\wedge \frac{\pi^d(Y^d)}{\pi^d(X^d)}\right)\right]-\frac{2\ell^2}{d-1} \Phi\left(-\frac{\ell \sqrt{I_d(X^d)}}{2}\right)\right|\\
&=\EE_{X^d\sim \pi^d}\left|\EE_{Y_i}\left\{(Y_i-X_i)^2\left[\EE_{Y_{-i}}\left(1\wedge \frac{\pi^d(Y^d)}{\pi^d(X^d)}\right)-2 \Phi\left(-\frac{\ell \sqrt{I_d(X^d)}}{2}\right)\right]\right\}\right|\\
&=o(d^{-1}).
\]
It then suffices to show 
\[
&\sup_{x^d \in F_d}\left|\EE_{Y_i}\left\{(Y_i-x_i)^2\mathbbm{1}_{y^d(i)\in F_d^{(i)}}\left[\EE_{Y_{-i}}\left(1\wedge \frac{\pi^d(Y^d)}{\pi^d(x^d)}\right)-2 \Phi\left(-\frac{\ell \sqrt{I_d(x^d)}}{2}\right)\right]\right\}\right|\\
&\le \EE_{Y_i}\left\{(Y_i-x_i)^2 \sup_{y^d(i)\in F_d^{(i)}, x^d \in F_d}\left|\EE_{Y_{-i}}\left(1\wedge \frac{\pi^d(Y^d)}{\pi^d(x^d)}\right)-2 \Phi\left(-\frac{\ell \sqrt{I_d(x^d)}}{2}\right)\right|\right\}=o(d^{-1}),
\]
where $y^d(i):=(x_1,\dots,x_{i-1},Y_i,x_{i+1},\dots,x_d)$. 
Defining $M_{x^d}^{(i)}(Y_i):=\EE_{Y_{-i}}\left(1\wedge \frac{\pi^d(Y^d)}{\pi^d(x^d)}\right)$,
since
\[
\log\frac{\pi^d(Y^d)}{\pi^d(x^d)}&=\log\frac{\pi_i(Y_i)}{\pi_i(x_i)}+\log\frac{\pi_{-i}(Y_{-i}\,|\,Y_i)}{\pi_{-i}(x_{-i}\,|\,x_i)}\\
&=\left(\log\frac{\pi_i(Y_i)}{\pi_i(x_i)}+\log\frac{\pi_{-i}(x_{-i}\,|\,Y_i)}{\pi_{-i}(x_{-i}\,|\,x_i)}\right)
+\log\frac{\pi_{-i}(Y_{-i}\,|\,Y_i)}{\pi_{-i}(x_{-i}\,|\,Y_i)},
\]
we can write
\[\label{tmp_eq_def_M}
M_{x^d}^{(i)}(Y_i)=&\EE_{Y_{-i}}\left[1\wedge \frac{\pi^d(Y^d)}{\pi^d(x^d)}\right]=\EE_{Y_{-i}}\left[1\wedge \exp\left(\log\frac{\pi^d(Y^d)}{\pi^d(x^d)}\right)\right]\\
=&\EE_{Y_{-i}}\left[1\wedge \exp\left(\log\frac{\pi_i(Y_i)}{\pi_i(x_i)}+\log\frac{\pi_{-i}(x_{-i}\,|\,Y_i)}{\pi_{-i}(x_{-i}\,|\,x_i)}+\log\frac{\pi_{-i}(Y_{-i}\,|\,Y_i)}{\pi_{-i}(x_{-i}\,|\,Y_i)}\right)\right].
\]
Note that the expectation is taken over $Y_{-i}$ and only the last term,
$\log\frac{\pi_{-i}(Y_{-i}\,|\,Y_i)}{\pi_{-i}(x_{-i}\,|\,Y_i)}$, involves $Y_{-i}$. 

In the following, we then first focus on approximating $\log\frac{\pi_{-i}(Y_{-i}\,|\,x_i)}{\pi_{-i}(x_{-i}\,|\,x_i)}$ for given $x^d\in F_d^+$. Since $Y^d\sim\mathcal{N}(x^d,\frac{\ell^2}{d-1}I)$, we first approximate $\log\frac{\pi_{-i}(Y_{-i}\,|\,x_i)}{\pi_{-i}(x_{-i}\,|\,x_i)}$ by the first two terms of its Taylor expansion.

Define
\[
m_1^{(i)}(Y_{-i},x^d):=(\nabla\log\pi_{-i})^T(Y_{-i}-x_{-i}) +\frac{1}{2}(Y_{-i}-x_{-i})^T[\nabla^2\log\pi_{-i}](Y_{-i}-x_{-i}),
\]
where
\[
(\nabla\log\pi_{-i})^T(Y_{-i}-x_{-i}):=\sum_{j\in\{1,\dots,d\},j\neq i} \frac{\partial \log \pi_{-i}(x_{-i}\,|\,x_i)}{\partial x_j}(Y_j-x_j)
\]
and $[\nabla^2\log\pi_{-i}]$ denotes the $(d-1)\times (d-1)$ matrix with elements $$\left\{\frac{\partial^2 \log \pi_{-i}(x_{-i}\,|\,x_i)}{\partial x_j\partial x_k}\right\}_{j,k\in\{1,\dots,d\},j\neq i, k\neq i}.$$ Then, we have the following result.
\begin{lemma}\label{lemma_approx1}
	Uniformly over $i\in\{1,\dots,d\}$, we have
	\[
	\sup_{x^d\in F_d^+}\EE_{Y_{-i}}\left[\left|m_1^{(i)}(Y_{-i},x^d)-\log\frac{\pi_{-i}(Y_{-i}\,|\,x_i)}{\pi_{-i}(x_{-i}\,|\,x_i)}\right|\right]\to 0.
	\]
\end{lemma}
\begin{proof}
	See \cref{proof_lemma_approx1}.
\end{proof}

Next, we approximate the second order term of the Taylor approximation $\frac{1}{2}(Y_{-i}-x_{-i})^T[\nabla^2\log\pi_{-i}](Y_{-i}-x_{-i})$ by a non-random term $\frac{1}{2}\frac{\ell^2}{d-1}\sum_{j\neq i}\frac{\partial^2 \log\pi_{-i}}{\partial x_j^2}$.
\begin{lemma} \label{lemma_approx2}
	Uniformly over $i\in\{1,\dots,d\}$, we have
	\[
	\sup_{x^d\in F_d^+}\EE_{Y_{-i}}\left[\left|(Y_{-i}-x_{-i})^T[\nabla^2\log\pi_{-i}](Y_{-i}-x_{-i})-\frac{\ell^2}{d-1}\sum_{j\neq i}\frac{\partial^2 \log\pi_{-i}}{\partial x_j^2}\right|\right]\to 0.
	\]
\end{lemma}
\begin{proof}
	See \cref{proof_lemma_approx2}.
\end{proof}
Defining
\[
m_2^{(i)}(Y_{-i},x^d):=(\nabla\log\pi_{-i})^T(Y_{-i}-x_{-i}) +\frac{1}{2}\frac{\ell^2}{d-1}\sum_{j\neq i}\frac{\partial^2 \log\pi_{-i}}{\partial x_j^2},
\]
we have
\[
m_2^{(i)}(Y_{-i},x^d)\sim \mathcal{N}\left(\ell^2S_d^{(i)}/2,\ell^2R_d^{(i)}\right),
\]
where
\[
R_d^{(i)}:=\frac{1}{d-1}\sum_{j\neq i} \left(\frac{\partial \log \pi_{-i}(x_{-i}\,|\,x_i)}{\partial x_j}\right)^2,\quad S_d^{(i)}:=\frac{1}{d-1}\sum_{j\neq i}\frac{\partial^2 \log\pi_{-i}(x_{-i}\,|\,x_i)}{\partial x_j^2}.
\]	
Next, we show we can approximate $S_d^{(i)}$ by $-R_d^{(i)}$.
\begin{lemma}\label{lemma_approx3}
	There exists a sequence of subsets of states $\{F_d'\}$, such that $\pi^d(F_d')\to 1$ and
	\[
	\sup_{i \in \{1,\dots,d\}} \sup_{x^d \in F_d'}\left| R_d^{(i)}+S_d^{(i)} \right|\to 0.
	\]
\end{lemma} 
\begin{proof}
	See \cref{proof_lemma_approx3}.
\end{proof}
Now defining
\[
m_3^{(i)}(Y_{-i},x^d):=(\nabla\log\pi_{-i})^T(Y_{-i}-x_{-i}) +\frac{1}{2}\frac{\ell^2}{d-1}\sum_{j\neq i} \left(\frac{\partial \log \pi_{-i}(x_{-i}\,|\,x_i)}{\partial x_j}\right)^2,
\]
we have
\[
m_3^{(i)}(Y_{-i},x^d)\sim \mathcal{N}\left(-\ell^2R_d^{(i)}/2,\ell^2R_d^{(i)}\right).
\]
By triangle inequality, we can write
\[
\left|m_3^{(i)}(Y_{-i},x^d)-\log\frac{\pi_{-i}(Y_{-i}\,|\,x_i)}{\pi_{-i}(x_{-i}\,|\,x_i)}\right|&\le \left|m_1^{(i)}(Y_{-i},x^d)-\log\frac{\pi_{-i}(Y_{-i}\,|\,x_i)}{\pi_{-i}(x_{-i}\,|\,x_i)}\right|\\
&+ \left| m_2^{(i)}(Y_{-i},x^d)-m_1^{(i)}(Y_{-i},x^d) \right|\\
&+ \left| m_3^{(i)}(Y_{-i},x^d)-m_2^{(i)}(Y_{-i},x^d) \right|.
\]
Therefore, using \cref{lemma_approx1,lemma_approx2,lemma_approx3}, we get
\[\label{tmp_eq1}
\sup_{i\in\{1,\dots,d\}}\sup_{x^d\in F_d^+\cap F_d'}\EE_{Y_{-i}}\left[\left|m_3^{(i)}(Y_{-i},x^d)-\log\frac{\pi_{-i}(Y_{-i}\,|\,x_i)}{\pi_{-i}(x_{-i}\,|\,x_i)}\right|\right]\to 0.
\]
Next, we abuse the notation a little bit by defining
\[
R_d^{(i)}(y):=\frac{1}{d-1}\sum_{j\neq i} \left(\frac{\partial \log \pi_{-i}(x_{-i}\,|\,x_i=y)}{\partial x_j}\right)^2.
\]
Then by the definition of $m_3^{(i)}$, we replace $x^d$ by $y^d(i)=(x_1,\dots,x_{i-1},Y_i,x_{i+1},\dots,x_d)$, which yields
\[
m_3^{(i)}(Y_{-i},y^d(i))&=
(\nabla\log\pi_{-i}(x_{-i}\,|\,Y_i))^T(Y_{-i}-x_{-i})\\ &\quad +\frac{1}{2}\frac{\ell^2}{d-1}\sum_{j\neq i} \left(\frac{\partial \log \pi_{-i}(x_{-i}\,|\,Y_i)}{\partial x_j}\right)^2.
\]
Then, we have
\[
m_3^{(i)}(Y_{-i},y^d(i))\sim \mathcal{N}\left(-\ell^2R_d^{(i)}(Y_i)/2,\ell^2R_d^{(i)}(Y_i)\right).
\]
Recall that $M_{x^d}^{(i)}(Y_i)=\EE_{Y_{-i}}\left[1\wedge \exp\left(\log\frac{\pi_i(Y_i)}{\pi_i(x_i)}+\log\frac{\pi_{-i}(x_{-i}\,|\,Y_i)}{\pi_{-i}(x_{-i}\,|\,x_i)}+\log\frac{\pi_{-i}(Y_{-i}\,|\,Y_i)}{\pi_{-i}(x_{-i}\,|\,Y_i)}\right)\right]$, defining
\[
\hat{M}_{x^d}^{(i)}(Y_i)=\EE_{Y_{-i}}\left[1\wedge \exp\left(\log\frac{\pi_i(Y_i)}{\pi_i(x_i)}+\log\frac{\pi_{-i}(x_{-i}\,|\,Y_i)}{\pi_{-i}(x_{-i}\,|\,x_i)}+ m_3^{(i)}(Y_{-i},y^d(i))\right)\right],
\]
we next apply the following two lemmas from \cite{Roberts1997}. 
\begin{lemma}(\cite[Proposition 2.2]{Roberts1997})\label{lemma_Lipschitz_bound}
	The function $g(x)=1\wedge {e}^{x}$ is Lipschitz such that 
	\[
	|g(x)-g(y)|\le |x-y|,\quad \forall x, y.
	\]
\end{lemma}
\begin{lemma}(\cite[Proposition 2.4]{Roberts1997})\label{lemma_Lipschitz_expectation}
	If $z\sim \normal(\mu,\sigma^2)$ then
	\[
	\EE(1\wedge e^z)=\Phi(\mu/\sigma)+\exp(\mu+\sigma^2/2)\Phi(-\sigma-\mu/\sigma).
	\]
\end{lemma}
By \cref{lemma_Lipschitz_bound} and \cref{tmp_eq1}, we have that uniformly over $i\in\{1,\dots,d\}$
\[\label{tmp_eq_M_by_hatM}
\sup_{y^d(i)\in F_d^+\cap F_d'}\left| M_{x^d}^{(i)}(Y_i)- \hat{M}_{x^d}^{(i)}(Y_i)\right|\to 0.
\]
Applying \cref{lemma_Lipschitz_expectation} to $\hat{M}_{x^d}^{(i)}(Y_i)$ yields
\[
\hat{M}_{x^d}^{(i)}(Y_i)
&=\Phi\left(R_d^{(i)}(Y_i)^{-1/2}\left(\ell^{-1}\log\frac{\pi^d(y^d(i))}{\pi^d(x^d)}-\ell R_d^{(i)}(Y_i)/2\right)\right)\\
&\quad +\exp\left(\log\frac{\pi^d(y^d(i))}{\pi^d(x^d)}\right)\Phi\left(-\ell R_d^{(i)}(Y_i)^{1/2}/2-\log\frac{\pi^d(y^d(i))}{\pi^d(x^d)} R_d^{(i)}(Y_i)^{-1/2}\ell^{-1}\right).
\]
Note that it is easy to check that
$\hat{M}_{x^d}^{(i)}(x_i)=2\Phi\left(-\frac{\ell \sqrt{R_d^{(i)}}}{2}\right)$.
We then show $\hat{M}_{x^d}^{(i)}(x_i)$ converges to $2 \Phi\left( -\frac{\ell \sqrt{I_d(x^d)}}{2}\right)$.
\begin{lemma}\label{lemma_approx4}
	\[
	\sup_{i\in\{1,\dots,d\}}\sup_{x^d \in F_d^+}\left|2\Phi\left(-\frac{\ell \sqrt{R_d^{(i)}}}{2}\right) -2 \Phi\left( -\frac{\ell \sqrt{I_d(x^d)}}{2}\right)\right|\to 0.
	\]
\end{lemma}
\begin{proof}
	See \cref{proof_lemma_approx4}.
\end{proof}
Finally, using Taylor expansion together with $\EE_{Y_i}(Y_i-x_i)^2=\ell^2/(d-1)$ and $\EE_{Y_i}|Y_i-x_i|^3=\bigO(d^{-3/2})$, we have 
\[
&\EE_{Y_i}\left\{(Y_i-x_i)^2 \sup_{y^d(i)\in F_d^+}\left|\hat{M}_{x^d}^{(i)}(Y_i)-2 \Phi\left(-\frac{\ell \sqrt{I_d(x^d)}}{2}\right)\right|\right\}\\
&\le \frac{\ell^2}{d-1} \sup_{x^d \in F_d^+}\left|2\Phi\left(-\frac{\ell \sqrt{R_d^{(i)}}}{2}\right) -2 \Phi\left( -\frac{\ell \sqrt{I_d(x^d)}}{2}\right)\right|\\
&\quad + \bigO(d^{-3/2}) \sup_{y^d(i)\in F_d^+} \left|\frac{\dee \hat{M}_{x^d}^{(i)}(y_i) }{\dee y_i}(Y_i) \right|.
\]
For the last term, we have the following lemma.
\begin{lemma}\label{lemma_approx5}
	\[
	\sup_{i\in \{1,\dots,d\}}\sup_{y^d(i)\in F_d^+} \left|\frac{\dee \hat{M}_{x^d}^{(i)}(y_i) }{\dee y_i}(Y_i)\right|=o\left(d^{1/2}\right).
	\]
\end{lemma}
\begin{proof}
	See \cref{proof_lemma_approx5}.
\end{proof}
The proof of \cref{key_theorem} is completed by applying \cref{lemma_approx4} and  \cref{lemma_approx5}.

\section{Proof of Lemmas in \cref{proof_key_theorem}}\label{proof_lemmas_key_theorem}
\subsection{Proof of \cref{lemma_approx1}}\label{proof_lemma_approx1}

{For $x^d\in F_d^+$,}
by Taylor expansion and mean value theorem, we have 
\[
&|\log \pi_{-i}(Y_{-i}\,|\,x_i)-\log \pi_{-i}(x_{-i}\,|\,x_i)-m_1(Y_{-i},x^d)|\\
&\le \sup_{\tilde{x}^d{\in \Reals^d}}\left|\frac{1}{6}\sum_{j,k,l\neq i}\frac{\partial^3 \log \pi^d(\tilde{x}^d)}{\partial x_j\partial x_k\partial x_l}(Y_j-x_j)(Y_k-x_k)(Y_l-x_l)	\right|.
\]
In the above summation, the summation over {the cases of $j=k=l$ equals to}
\[
{\sup_{\tilde{x}^d\in \Reals^d}\left|\frac{\partial^3\log \pi^d(\tilde{x}^d)}{\partial x_j^3}\right|}\bigO(d\EE |Y_j-x_j|^3)={o(d^{1/2})}\bigO\left(d(\sqrt{\ell^2/(d-1)})^3\right)=o(1).
\]
For the cases of $j=k\neq l$, we have
\[
\sum_{j=k\neq l}\frac{\partial^3 \log \pi_{-i}(\tilde{x}^d)}{\partial x_j^2 \partial x_l}(Y_j-x_j)^2(Y_l-x_l)
=\sum_j (Y_j-x_j)^2 \sum_{l\neq k} \frac{\partial^3 \log \pi_{-i}(\tilde{x}^d)}{\partial x_j^2 \partial x_l} (Y_l-x_l).
\]
By Assumption \ref{assumption_A3}, we have
$\EE\left|\sum_{j\neq l} \frac{\partial^3 \log \pi_{-i}(\tilde{x}^d)}{\partial x_j^2 \partial x_l} (Y_l-x_l)\right|=\bigO(l_d/d) o(d/l_d)=o(1)$ since $\frac{\partial^3 \log \pi^d(\tilde{x}^d)}{\partial x_j^2 \partial x_l}$ goes to zero when $|k-i|>l_d$. Then, by $\EE |Y_j-x_j|^2=\bigO(1/d)$,  the summation over all cases of $j=k\neq l$ equals to $d\bigO_{\Pr}(1/d)o_{\Pr}(1)=o_{\Pr}(1)$.

Finally, for $j\neq k\neq l$, it suffices to show
\[
&\sup_{\tilde{x}^d\in {\Reals^d}}\left|\sum_{j\neq k\neq l\neq i}\frac{\partial^3 \log \pi_{-i}(\tilde{x}^d)}{\partial x_j\partial x_k\partial x_l}(Y_j-x_j)(Y_k-x_k)(Y_l-x_l)	\right|\\
&\le \sum_{i\neq j\neq k\neq l}\left(\sup_{\tilde{x}^d\in {\Reals^d}}\left|\frac{\partial^3 \log \pi_{-i}(\tilde{x}^d)}{\partial x_j\partial x_k \partial x_l}\right|\right)\left|(Y_j-x_j)(Y_k-x_k)(Y_l-x_l)\right|=o_{\Pr}(1).
\]
Note that {$\{\left|(Y_j-x_j)(Y_k-x_k)(Y_l-x_l)\right|\}_{j\neq k\neq l}$ are independent random variables which don't depend on the values of $x_j,x_k,x_l$, and 
	\[
	\left|(Y_j-x_j)(Y_k-x_k)(Y_l-x_l)\right|=\bigO_{\Pr}\left((\sqrt{\ell^2/(d-1)})^3\right)=\bigO_{\Pr}(d^{-3/2}).
	\]}
Therefore, {the summation for cases $j\neq k\neq l$ is $o_{\Pr}(1)$ under Assumption \ref{assumption_A3}.}
We have proven the result for fixed $i$. Finally, it is easy to check the proof holds uniformly over $i\in\{1,\dots,d\}$.

\subsection{Proof of \cref{lemma_approx2}}\label{proof_lemma_approx2}

\begin{lemma}(Quadratic Form of Gaussian Random Vector)\label{lemma_quadratic_form}
	If $z^d\sim \normal_d(\mu,\Sigma)$, then
	\[
	\EE(z^T A z)=\trace(A\Sigma)+\mu^T A\mu,\quad \var(z^T Az)=2\trace(A\Sigma A\Sigma)+4\mu^T A\Sigma A\mu.
	\]
\end{lemma}
Note that $Y_{-i}\sim \normal_{d-1}(x_{-i},\frac{\ell^2}{d-1}I)$ and  $(Y_{-i}-x_{-i})^T[\nabla^2\log\pi_{-i}](Y_{-i}-x_{-i})$ is a quadratic form of Gaussian random vector.	By \cref{lemma_quadratic_form},\[
\EE\left[(Y_{-i}-x_{-i})^T[\nabla^2\log\pi_{-i}](Y_{-i}-x_{-i})\right]=\frac{\ell^2}{d-1}\sum_{j\neq i}\frac{\partial^2 \log\pi_{-i}}{\partial x_j^2}.
\]
Therefore, it suffices to show the variance of the quadratic form goes to zero.
Using the assumptions, the variance satisfies
\[
&\frac{2\ell^4}{(d-1)^2}\trace\left([\nabla^2\log\pi_{-i}][\nabla^2\log\pi_{-i}]\right)\\
&=\frac{2\ell^4}{(d-1)^2}\sum_{j\neq i}\sum_{k\neq i} \left(\frac{\partial^2 \log \pi_{-i}}{\partial x_j\partial x_k}\right)^2\\
&\le\frac{2\ell^4}{(d-1)^2}\sum_{l=0}^{d-1}\sum_{\{j,k: |j-k|=l\}} \left(\frac{\partial^2 \log \pi^d}{\partial x_j\partial x_k}\right)^2\\
&=\frac{2\ell^4}{(d-1)^2}\sum_{l\le l_d}\sum_{\{j,k: |j-k|=l\}} \left(\frac{\partial^2 \log \pi^d}{\partial x_i\partial x_j}\right)^2\\
&\qquad+\frac{2\ell^4}{(d-1)^2}\sum_{l> l_d}\sum_{\{j,k: |j-k|=l\}} \left(\frac{\partial^2 \log \pi^d}{\partial x_j\partial x_k}\right)^2\\
&\le\frac{2\ell^4}{(d-1)^2}(d-1)l_d\sup_{|j-k|\le l_d}\sup_{x^d\in F_d^+} \left(\frac{\partial^2 \log \pi^d}{\partial x_j\partial x_k}\right)^2\\
&\qquad+\frac{2\ell^4}{(d-1)^2}(d-1)^2\sup_{|j-k|>l_d} \sup_{x^d\in F_d^+}\left(\frac{\partial^2 \log \pi^d}{\partial x_j\partial x_k}\right)^2\\
&=\mathcal{O}(l_d/d)o(d/l_d)+o(1)=o(1),
\]
where we have used 
$\sup_{x^d\in F_d^+}\sup_{|j-k|\le l_d} \frac{\partial^2 \log \pi^d}{\partial x_j \partial x_k}=o(\sqrt{d/l_d})$ from Assumption \ref{assumption_A1}.

\subsection{Proof of \cref{lemma_approx3}}\label{proof_lemma_approx3}
Note that
\[
R_d^{(i)}+S_d^{(i)}&=\frac{1}{d-1}\sum_{j\neq i} \left(\frac{\partial \log \pi_{-i}}{\partial x_j}\right)^2+\frac{1}{d-1}\sum_{j\neq i}\frac{\partial^2 \log\pi_{-i}}{\partial x_j^2}\\
&=\frac{1}{d-1}\sum_{j\neq i}\left\{ \left(\frac{\partial \log \pi^d}{\partial x_j}\right)^2+\frac{\partial^2 \log\pi^d}{\partial x_j^2}\right\}\\
&=\frac{1}{d-1}\sum_{j\neq i}\left\{\frac{1}{(\pi^d)^2}\left(\frac{\partial \pi^d}{\partial x_j}\right)^2+\frac{\partial }{\partial x_j} \left(\frac{\partial \log \pi^d}{\partial x_j}\right)\right\}\\
&=\frac{1}{d-1}\sum_{j\neq i}\left\{\frac{1}{(\pi^d)^2}\left(\frac{\partial \pi^d}{\partial x_j}\right)^2+\frac{\partial }{\partial x_j} \left(\frac{1}{\pi^d}\frac{\partial  \pi^d}{\partial x_j}\right)\right\}\\
&=\frac{1}{d-1}\sum_{j\neq i}\left\{\frac{1}{(\pi^d)^2}\left(\frac{\partial \pi^d}{\partial x_j}\right)^2+\frac{\pi^d\frac{\partial^2 \pi^d}{\partial x_j^2}-\left(\frac{\partial \pi^d}{\partial x_j}\right)^2}{(\pi^d)^2}\right\}\\
&=\frac{1}{(d-1)}\sum_{j\neq i} \frac{\partial^2 \pi^d}{\partial x_j^2}\frac{1}{\pi^d}.
\]
Next, we show $\EE\left[\sup_i (R_d^{(i)}+S_d^{(i)})^2\right]$ converges to $0$.
To prove this, consider writing $\EE\left[\sup_i(R_d^{(i)}+S_d^{(i)})^2\right]$ as sum of $(d-1)^2$ terms
{\[
&\EE\left[\sup_i(R_d^{(i)}+S_d^{(i)})^2\right]=\frac{1}{(d-1)^2}\int\sup_i\sum_{j\neq i}\sum_{k\neq i} \left(\frac{\partial^2 \pi^d}{\partial x_j^2}\frac{1}{\pi^d}\right)\left(\frac{\partial^2 \pi^d}{\partial x_k^2}\frac{1}{\pi^d}\right)\pi^d \dee x^d\\
&\le \frac{1}{(d-1)^2}\sum_{j=1}^d\sum_{k=1}^d\int \left(\frac{\partial^2 \pi^d}{\partial x_j^2}\frac{1}{\pi^d}\right)\left(\frac{\partial^2 \pi^d}{\partial x_k^2}\frac{1}{\pi^d}\right)\pi^d \dee x^d -\frac{2}{(d-1)^2}\int \inf_i \sum_{j\neq i} \left(\frac{\partial^2 \pi^d}{\partial x_j^2}\frac{1}{\pi^d}\right) \pi^d \dee x^d\\
&=\frac{1}{(d-1)^2}\sum_{j=1}^d\sum_{k=1}^d\int \left(\frac{\partial^2 \pi^d}{\partial x_j^2}\frac{1}{\pi^d}\right)\left(\frac{\partial^2 \pi^d}{\partial x_k^2}\frac{1}{\pi^d}\right)\pi^d \dee x^d +o(1),
\]
where the last equality follows from
\[
\frac{2}{(d-1)^2}\int \inf_i \sum_{j\neq i} \left(\frac{\partial^2 \pi^d}{\partial x_j^2}\frac{1}{\pi^d}\right) \pi^d \dee x^d&\ge \frac{2}{(d-1)^2}\int \inf_i \sum_{j\neq i} \left(\frac{\partial^2\log \pi_{-i}}{\partial x_j^2}\right) \pi^d \dee x^d\\
&=\frac{2}{(d-1)^2}  o(d\sqrt{d/l_d})=o(1).
\]
}

When $|j-k|\ge l_d$, by Assumption \ref{assumption_A2}, we have
\[
&\int \left(\frac{\partial^2 \pi^d}{\partial x_j^2}\frac{1}{\pi^d}\right)\left(\frac{\partial^2 \pi^d}{\partial x_k^2}\frac{1}{\pi^d}\right)\pi^d \dee x^d\\
&=\int \left(\frac{\partial^2 \pi^d}{\partial x_j^2}\right)\left(\frac{\partial^2 \pi^d}{\partial x_k^2}\right)\frac{1}{\pi^d} \dee x^d\\
&=\int \left(\frac{\partial^2 \pi_{j,k|-j-k}}{\partial x_j^2}\right)\left(\frac{\partial^2 \pi_{j,k|-j-k}}{\partial x_k^2}\right)\frac{1}{\pi_{j,k|-j-k}}\pi_{-j-k} \dee x_{-j-k} \dee x_j \dee x_k\\
&\le \int\left[\sup_{x^d\in F_d}\int \left(\frac{\partial^2 \pi_{j,k|-j-k}}{\partial x_j^2}\right)\left(\frac{\partial^2 \pi_{j,k|-j-k}}{\partial x_k^2}\right)\frac{1}{\pi_{j,k|-j-k}}\dee x_j\dee x_k\right] \pi_{-j-k}\dee x_{-j-k}\\
&\to 0.
\]
This implies
{$\EE\left[\sup_i (R_d^{(i)}+S_d^{(i)})^2\right]=\frac{\mathcal{O}(d\, l_d) + (d-l_d)^2 o(1)}{(d-1)^2}+o(1)\to 0$.} Therefore, uniformly over $i$, $R_d^{(i)}+S_d^{(i)}\to 0$ in probability, then there exists a sequence $\{F_d'\}$ such that $\Pr(R_d^{(i)}+S_d^{(i)} \in F_d', \forall i)\to 1$  and the following holds
\[
\sup_i \sup_{x^d \in F_d'}  \left|R_d^{(i)}+S_d^{(i)}\right|\to 0.
\]
\subsection{Proof of \cref{lemma_approx4}}\label{proof_lemma_approx4}
Note that Assumption \ref{assumption_A4} implies 
\[
\sup_{i\in\{1,\dots,d\}} \sup_{x^d \in F_d^+}\frac{\partial}{\partial x_i}\log \pi^d(x^d)=o\left(d^{1/2}\right).
\]
Then, by the definitions of $R_d^{(i)}$ and $I_d(x^d)$, we have
\[
R_d^{(i)}-I_d(x^d)&=\frac{1}{d-1}\sum_{j\neq i} \left(\frac{\partial \log \pi_{-i}(x_{-i}\,|\,x_i)}{\partial x_j}\right)^2-\frac{1}{d}\sum_{j=1}^d \left(\frac{\partial}{\partial x_j}\log \pi^d(x^d) \right)^2\\
&=\frac{1}{d-1}\sum_{j\neq i} \left(\frac{\partial \log \pi^d(x^d)}{\partial x_j}\right)^2-\frac{1}{d}\sum_{j=1}^d \left(\frac{\partial}{\partial x_j}\log \pi^d(x^d) \right)^2\\
&=\frac{1}{d}R_d^{(i)}-\frac{1}{d}\left(\frac{\partial}{\partial x_i}\log \pi^d(x^d) \right)^2\to 0.
\] 
\subsection{Proof of \cref{lemma_approx5}}\label{proof_lemma_approx5}
Recall that we have shown 
\[
\hat{M}_{x^d}^{(i)}(Y_i)
&=\Phi\left(R_d^{(i)}(Y_i)^{-1/2}\left(\ell^{-1}\log\frac{\pi^d(y^d(i))}{\pi^d(x^d)}-\ell R_d^{(i)}(Y_i)/2\right)\right)\\
&\quad +\exp\left(\log\frac{\pi^d(y^d(i))}{\pi^d(x^d)}\right)\Phi\left(-\ell R_d^{(i)}(Y_i)^{1/2}/2-\log\frac{\pi^d(y^d(i))}{\pi^d(x^d)} R_d^{(i)}(Y_i)^{-1/2}\ell^{-1}\right).
\]
For notational simplicity, we omit the index $i$ and write $R_d^{(i)}$ by $R_d$. To simplify the derivation, we note that $\hat{M}_{x^d}^{(i)}(y)$ has the following form
\[
M(y)=\Phi\left(f(y)g(y)-\frac{1}{2}f^{-1}(y)\right)+\exp(g(y))\Phi\left(-\frac{1}{2}f^{-1}(y)-f(y)g(y)\right),
\]
where $f^{-1}(y):=\ell R_d^{1/2}(y)$ and $g(y)=\log \pi^d(y^d(i))-\log \pi^d(x^d)$. Taking the derivative with respect to $y$, we get
\[
\frac{\dee M(y)}{\dee y}
&= \Phi'(fg-f^{-1}/2)\frac{\dee }{\dee y}(fg-f^{-1}/2)\\
&\quad +\exp(g)\Phi'(-f^{-1}/2-fg)\frac{\dee}{\dee y}(-fg-f^{-1}/2)\\
&\quad +\exp(g) \left(\frac{\dee }{\dee y}g\right)\Phi(-fg-f^{-1}/2)\\
&\le \|\Phi'\|_{\infty} \left|\frac{\dee f}{\dee y}g+\frac{\dee g}{\dee y}f -\frac{1}{2}\frac{\dee f^{-1}}{\dee y}\right|\\
&\quad +\exp(g) \|\Phi'\|_{\infty}\left|\frac{\dee f}{\dee y}g+\frac{\dee g}{\dee y}f +\frac{1}{2}\frac{\dee f^{-1}}{\dee y}\right|\\
&\quad +\exp(g) \left|\frac{\dee g}{\dee y}\right|\|\Phi\|_{\infty}
\]

Note that both $\Phi$ and $\Phi'$ are bounded functions.
It then suffices to show
\[
&\exp(g)\left|\frac{\dee g}{\dee y}\right|=o(d^{1/2}),\quad \exp(g)\left|\frac{\dee f}{\dee y} g\right|=o(d^{1/2}),\\
&\exp(g)\left|\frac{\dee g}{\dee y} f\right|=o(d^{1/2}),\quad \exp(g)\left|\frac{\dee f^{-1}}{\dee y} \right|=o(d^{1/2}).
\]
Observing that $\frac{\dee f^{-1}}{\dee y}=\frac{1}{2}\ell R_d'/R_d^{1/2}$ and $\frac{\dee f}{\dee y}=-\frac{1}{2\ell} \frac{1}{R_d}\frac{R_d'}{R_d^{1/2}}$, if we can show
\[\label{tmp_eq2}
\sup_{i\in\{1,\dots,d\}} \frac{\dee R_d^{(i)}(y)}{\dee y}\frac{1}{[R_d^{(i)}(y)]^{1/2}}=o(1),
\]
then we can get $\frac{\dee f^{-1}}{\dee y}=o(1)$ and $\frac{\dee f}{\dee y}=o(1/R_d)$. Using $R_d^{(i)}\to I_d(x^d)$ from \cref{proof_lemma_approx4}, it suffices to show
\[
&\left(\sup_{x^d \in F_d^+} \pi^d(x^d)\right) \left(\sup_i\sup_{x^d \in F_d^{(i)}} \frac{\partial \log \pi^d(x^d)}{\partial x_i}\right) =o(d^{1/2}),\\
&\left(\sup_{x^d \in F_d^+} \pi^d(x^d)\right) \left(\sup_{x^d \in F_d^+} \left|\log(\pi^d(x^d))/I_d(x^d)\right|\right)=o(d^{1/2}),\\
&\left(\sup_{x^d \in F_d^+}\pi^d(x^d)\right)\left(\sup_i\sup_{x^d \in F_d^{(i)}}\left|\frac{\partial \log \pi^d(x^d)}{\partial x_i}/\sqrt{I_d(x^d)} \right|\right)=o(d^{1/2}).
\]
One can easily verify that the above equations hold under Assumption \ref{assumption_A4}.

Finally, we complete the proof by showing \cref{tmp_eq2}. Recall that
\[
R_d^{(i)}(y)=\frac{1}{d-1}\sum_{j\neq i} \left(\frac{\partial \log \pi_{-i}(x_{-i}\,|\,x_i=y)}{\partial x_j}\right)^2.
\]
For notational simplicity, we write
\[
R_d^{(i)}(y)=\frac{1}{d-1}\sum_{j\neq i} f_j^2(y),
\]
where $f_j(y):=\frac{\partial \log \pi_{-i}(x_{-i}\,|\,x_i=y)}{\partial x_j}$. Then,  by Cauchy--Schwartz inequality
\[
\frac{\partial R_d^{(i)}(y)}{\dee y}=\frac{2}{d-1}\sum_{j\neq i} f_j(y) f'_j(y)\le \frac{2}{d-1}\sqrt{\sum_{j\neq i} f_j^2(y)\sum_{j\neq i} |f_j'(y)|^2}.
\]
Note that by \ref{assumption_A1}, if $|i-j|>l_d$ then $f_j'(y)\le \sup_{x^d \in F_d}\frac{\partial^2 \log \pi^d(x^d)}{\partial x_i\partial x_j}\to 0$. Hence, we have
\[
&\sup_{i\in\{1,\dots,d\}} \frac{\dee R_d^{(i)}(y)}{\dee y}\frac{1}{[R_d^{(i)}(y)]^{1/2}}\le\sup_i \frac{\frac{2}{d-1}\sqrt{\sum_{j\neq i} f_j^2(y)\sum_{j\neq i} |f_j'(y)|^2}
}{\sqrt{\frac{1}{d-1}\sum_{j\neq i} f_j^2(y)}}\\
&=2\sup_i \sqrt{\frac{1}{d-1}\sum_{j\neq i}|f_j'(y)|^2 }
\le 2\sqrt{\frac{1}{d-1}\sum_{j=1}^d |f_j'(y)|^2 }=o\left(\sqrt{\frac{l_d }{d}(\sqrt{d/l_d})^2}\right)=o(1).
\]

\section{Proof of \cref{thm_diffusion}}\label{proof_thm_diffusion}

Similar to \cref{proof_key_theorem}, we assume the coordinates are linear ordered for simplicity. 
The proof follows the framework of \cite{Roberts1997} using the generator approach \cite{Ethier1986}.

{
	Define the (discrete time) generator of $x^d$ by
	\[
		(G_d f)(x^d):&=d \EE_{Y^d}\left\{[f(Y^d)-f(x^d)]\left(1\wedge \frac{\pi^d(Y^d)}{\pi^d(x^d)}\right)\right\},
	\]
	for any function $f$ for which this definition makes sense. In the Skorokhod topology, it doesn't cause any problem to treat $G_d$ as a continuous time generator. We shall restrict attention to test functions such that $f(x^d)=f(x_1)$.
}
{	
	We show uniform convergence of $G_d$ to $G$, the generator of the limiting (one-dimensional) Langevin diffusion, for a suitable large class of real-valued functions $f$, where, for some fixed function $h(\ell)$, 
	\[
		(Gf)(x_1):=h(\ell)\left\{\frac{1}{2}f''(x_1)+\frac{1}{2}[(\log \tilde{\pi})'(x_1)]\,f'(x_1)\right\},
	\]
	in which $\tilde{\pi}$ is a one-dimensional density of the first coordinate of $\pi^d$.} {Since we have assumed in \ref{assumption_A6} that $(\log \tilde{\pi})'$ is Lipschitz, by \cite[Thm 2.1 in Ch.8]{Ethier1986}, a core for the generator has domain $C_c^{\infty}$, which is the class of continuous functions with compact support such that all orders of derivatives exist. This enable us to restrict attentions to functions $f_c\in C_c^{\infty}$ such that $f_c(x^d)=f_c(x_1)$.
}


Note that using Assumption \ref{assumption_A2_plus}, and the assumption $\pi^d(F_d^c)=\bigO(d^{-1-\delta})$, following the arguments in the proof of \cref{lemma_approx3} we can get a stronger version of \cref{lemma_approx3} for $F_d':=\{x^d: \sup_i|R_d^{(i)}+S_d^{(i)}|\le d^{-\delta}\}$. Then using a union bound yields
\[
\Pr(X^d(\lfloor ds\rfloor) \notin F_d\cap F_d', \exists 0\le s \le t)\to 0.
\]
Therefore, for any fixed $t$, if $d\to\infty$ then the probability of all $X^d(\lfloor ds \rfloor), 0\le s\le t$ are in $F_d\cap F_d'$ goes to $1$. Since $F_d\cap F_d'\subseteq F_d^+\cap F_d'\subseteq F_d^+$, it suffices to consider $x^d\in F_d^+$.

Note that   $Y^d\sim\normal(x^d,\frac{\ell^2}{d-1}I)$, we can write
\[
(G_d f_c)(x^d)=d \EE_{Y_1}\left\{[f_c(Y_1)-f_c(x_1)]\EE_{Y_{-1}}\left[1\wedge \frac{\pi^d(Y^d)}{\pi^d(x^d)}\right]\right\},
\]
where $\EE_{Y_{-1}}[\cdot]$ is short for $\EE_{Y_2,\dots,Y_d\,|\,Y_1}[\cdot]$ and $\pi^d$ denotes the target distribution in $d$-dimension. The goal is then to prove $(G_d f_c)$ converges to $(G f_c)$.

Recall the definition \cref{tmp_eq_def_M}, we omit the index to write $M_{x^d}^{(1)}$ as $M_{x^d}$, which is defined by
\[
M_{x^d}(Y_1)=\EE_{Y_{-1}}\left(1\wedge \frac{\pi^d(Y^d)}{\pi^d(x^d)}\right).
\]
Then we have previously shown in \cref{tmp_eq_M_by_hatM} that $M_{x^d}(Y_1)$ can be approximated by
\[
\hat{M}_{x^d}(Y_1)
&=\Phi\left(R_d(Y_1)^{-1/2}\left(\ell^{-1}\log\frac{\pi^d(Y_1,x_{-1})}{\pi^d(x^d)}-\ell R_d(Y_1)/2\right)\right)\\
&\quad +\exp\left(\log\frac{\pi^d(Y_1,x_{-1})}{\pi^d(x^d)}\right)\Phi\left(-\ell R_d(Y_1)^{1/2}/2-\log\frac{\pi^d(Y_1,x_{-1})}{\pi^d(x^d)} R_d(Y_1)^{-1/2}\ell^{-1}\right)
\]
For $x^d\in F_d^+$, some properties of $\hat{M}_{x^d}$ is given as follows.
\begin{lemma}\label{lemma_property_hatM} For $\hat{M}_{x^d}$, we have 
	\[
	\hat{M}_{x^d}(x_1)&=2\Phi\left(-\frac{\ell R_d^{1/2}(x_1)}{2}\right),\\
	\hat{M}'_{x^d}(x_1)
	&=\Phi\left(-\frac{\ell R_d^{1/2}(x_1)}{2}\right)  \frac{\dee [\log \pi_1(x)+\log\pi_{-1}(x_{-1}\,|x)]}{\dee x}(x_1)+o(1),\\
	\hat{M}'_{x^d}(x_1)&=o(d^{1/2}),\quad \sup_{x^d\in F_d^+} \hat{M}''_{x^d}=o(d^{1/2}).
	\]
\end{lemma}
\begin{proof}
	See \cref{proof_lemma_property_hatM}.
\end{proof}
Since $f_c(Y_1)-f_c(x_1)$ is bounded, it suffices to show
\[
\EE_{Y_1}\left\{d[f_c(Y_1)-f_c(x_1)]\hat{M}_{x^d}(Y_1)\right\}\to (Gf_c)(x_1).
\]
Now using mean value theorem and Taylor expansion of $\EE_{Y_1}\left\{[f_c(Y_1)-f_c(x_1)]\hat{M}_{x^d}(Y_1)\right\}$ at $(Y_1-x_1)$ yields
\[
&[f_c(Y_1)-f_c(x_1)]\hat{M}_{x^d}(Y_1)\\
&= \left[f_c'(x_1)(Y_1-x_1)+\frac{1}{2}f_c''(x_1)(Y_1-x_1)^2+K(Y_1-x_1)^3\right]\\
&\qquad\cdot\left[\hat{M}_{x^d}(x_1)+ \hat{M}'_{x^d}(x_1)(Y_1-x_1)+\frac{1}{2}\hat{M}''_{x^d}(x')(Y_1-x_1)^2\right]\\
&=f_c'(x_1)\hat{M}_{x^d}(x_1)(Y_1-x_1)+\left[\frac{1}{2}f_c''(x_1)\hat{M}_{x^d}(x_1) + f_c'(x_1) \hat{M}'_{x^d}(x_1)\right](Y_1-x_1)^2\\
&\qquad+\left[K\hat{M}_{x^d}(x_1)+\frac{1}{2}f_c''(x_1)\hat{M}'_{x^d}(x_1)+\frac{1}{2}\hat{M}''_{x^d}(x')f_c'(x_1)\right](Y_1-x_1)^3\\
&\qquad+\left[\frac{1}{4}\hat{M}''_{x^d}(x')f_c''(x_1) +K\hat{M}'_{x^d}(x_1) \right](Y_1-x_1)^4+ \frac{1}{2}\hat{M}''_{x^d}(x')K(Y_1-x_1)^5,
\]
where $K$ is a constant since $f_c$ has bounded third derivative. Note that both $f_c'(x_1)$ and $f_c''(x_1)$ are bounded as well. Therefore, taking expectation over $Y_1$ and using $\hat{M}'_{x^d}(x_1)=o(d^{1/2}),\sup_{x^d} \hat{M}''_{x^d}=o(d^{1/2})$ in \cref{lemma_property_hatM}, we have
\[
\EE_{Y_1}\left\{[f_c(Y_1)-f_c(x_1)]\hat{M}_{x^d}(Y_1)\right\}
&=\left[\frac{1}{2}f_c''(x_1)\hat{M}_{x^d}(x_1) + f_c'(x_1) \hat{M}'_{x^d}(x_1)\right]\frac{\ell^2}{d-1}+o(d^{-1}).
\]
Finally, by Assumption \ref{assumption_A6}, we have
\[
&f_c'(x_1)\hat{M}'_{x^d}(x_1)+\frac{1}{2}f_c''(x_1)\hat{M}_{x^d}(x_1)\\
&=2\Phi\left(-\frac{\ell R_d^{1/2}(x_1)}{2}\right)\left(\frac{1}{2}f_c''(x_1)+\frac{1}{2}f_c'(x_1)\frac{\dee [\log \pi_1(x)+\log\pi_{-1}(x_{-1}\,|x)]}{\dee x}(x_1)\right)\\
&=2\Phi\left(-\frac{\ell R_d^{1/2}(x_1)}{2}\right)\left(\frac{1}{2}f_c''(x_1)+\frac{1}{2}f_c'(x_1)\frac{\dee \log\pi_{1|-1}(x\,|\,x_{-1})}{\dee x}(x_1)\right)\\
&\to 2\Phi\left(-\frac{\ell I(x^d)^{1/2}}{2}\right)\left(\frac{1}{2}f_c''(x_1)+\frac{1}{2}f_c'(x_1)\frac{\dee \log\tilde{\pi}(x)}{\dee x}(x_1)\right)\\
&\to 2\Phi\left(-\frac{\ell \bar{I}^{1/2}}{2}\right)\left(\frac{1}{2}f_c''(x_1)+\frac{1}{2}f_c'(x_1)\frac{\dee \log\tilde{\pi}(x)}{\dee x}(x_1)\right),
\]
which implies that $\EE_{Y_1}\left\{d[f_c(Y_1)-f_c(x_1)]\hat{M}_{x^d}(Y_1)\right\}\to (Gf_c)(x_1)$
where $h(\ell):=2\ell^2\Phi(-\ell\sqrt{\bar{I}}/2)$.

\subsection{Proof of \cref{lemma_property_hatM}}\label{proof_lemma_property_hatM}
The proof is quite tedious. In order to simplify the notations, we first introduce the following lemma.
\begin{lemma}\label{lemma_tmp}
	 For the function $M(y)$ defined by
	\[
	M(y)=\Phi\left(f(y)g(y)-\frac{1}{2}f^{-1}(y)\right)+e^{g(y)}\Phi\left(-\frac{1}{2}f^{-1}(y)-f(y)g(y)\right),
	\]
	we have
\[
\frac{\dee M(y)}{\dee y}
&= \Phi'(fg-f^{-1}/2)\frac{\dee }{\dee y}(fg-f^{-1}/2)\\
&\quad +e^g\Phi'(-f^{-1}/2-fg)\frac{\dee}{\dee y}(-fg-f^{-1}/2)\\
&\quad +e^g \left(\frac{\dee }{\dee y}g\right)\Phi(-fg-f^{-1}/2).
\]
\[
	&\frac{\dee^2 M(y)}{\dee y^2}=\Phi''(fg-f^{-1}/2)\left[\frac{\dee}{\dee y}(fg-f^{-1}/2)\right]^2+\Phi'(fg-f^{-1}/2)\frac{\dee^2}{\dee y^2}(fg-f^{-1}/2)\\
	&+e^g\left(\frac{\dee}{\dee y} g\right)\Phi'(-f^{-1}/2-fg)\frac{\dee}{\dee y}(-fg-f^{-1}/2)\\
	&+e^g\left\{\Phi''(-fg-f^{-1}/2)\left[\frac{\dee}{\dee y}(-fg-f^{-1}/2)\right]^2+\Phi'(-fg-f^{-1}/2)\frac{\dee^2}{\dee y^2}(-fg-f^{-1}/2)\right\}\\
	&+e^g\left(\frac{\dee}{\dee y}g\right)\Phi'(-fg-f^{-1}/2)\frac{\dee}{\dee y}(-fg-f^{-1}/2)\\
	&+\Phi(-fg-f^{-1}/2)\left[e^g\left(\frac{\dee^2}{\dee y^2}g\right)+e^g \left(\frac{\dee}{\dee y}g\right)^2 \right].
	\]
	Furthermore, if $g(x_1)=0$, then we have
	\[
	\frac{\dee M(y)}{\dee y}(x_1)&= \left(\Phi'(-f^{-1}/2)\frac{\dee }{\dee y}(fg-f^{-1}/2)\right.\\
	&\quad \left.+\Phi'(-f^{-1}/2)\frac{\dee}{\dee y}(-fg-f^{-1}/2)\right.\\
	&\left.\quad + \left(\frac{\dee }{\dee y}g\right)\Phi(-f^{-1}/2)\right)(x_1)\\
	&=\left(\Phi'(-f^{-1}/2)\frac{\dee }{\dee y}(-f^{-1}) + \left(\frac{\dee }{\dee y}g\right)\Phi(-f^{-1}/2)\right)(x_1)\\
	&=-\Phi'\left(-\frac{f^{-1}(x_1)}{2}\right)\frac{\dee f^{-1}(y)}{\dee y}(x_1) + \frac{\dee g(y)}{\dee y}(x_1)\Phi\left(-\frac{f^{-1}(x_1)}{2}\right).
	\]
\end{lemma}
\begin{remark}
	Let $g(y)=\log\frac{\pi^d(y,x_{-1})}{\pi^d(x^d)}$ and $f^{-1}(y)=\ell R_d^{1/2}(y)$ then $\hat{M}_{x^d}(y)=M(y)$.
\end{remark}
Now substituting $g(y)=\log\frac{\pi^d(y,x_{-1})}{\pi^d(x^d)}$ and $f^{-1}(y)=\ell R_d^{1/2}(y)$ to \cref{lemma_tmp}, we have
\[
\hat{M}_{x^d}(x_1)=2\Phi\left(-\frac{\ell R_d^{1/2}(x_1)}{2}\right),
\]
and
\[
\hat{M}'_{x^d}(x_1)=&\frac{\dee \hat{M}_{x^d}(y)}{\dee y}(x_1)\\
=&\Phi\left(-\frac{\ell R_d^{1/2}(x_1)}{2}\right)  \frac{\dee [\log \pi_1(x)+\log\pi_{-1}(x_{-1}\,|x)]}{\dee x}(x_1)\\
&\quad-\Phi'\left(-\frac{\ell R_d^{1/2}(x_1)}{2}\right)\frac{\ell}{2R_d^{1/2}(x_1)}R_d'(x_1).
\]
Since $\Phi'$ is bounded and by \cref{tmp_eq2}, $R_d'(x_1)/R_d^{1/2}(x_1)\to 0$, therefore 
\[
\Phi'\left(-\frac{\ell R_d^{1/2}(x_1)}{2}\right)\frac{\ell}{2R_d^{1/2}(x_1)}R_d'(x_1)=o(1).
\]
Also, $\hat{M}'_{x^d}(x_1)=o(d^{1/2})$ since $\frac{\partial \log \pi^d}{\partial x_i}=\bigO(d^{\alpha/2})=o(d^{1/2})$.

Now we prove $ \sup_{x^d} \hat{M}''_{x^d}=o(d^{1/2})$.
For simplicity, we keep the notations of $f$ and $g$ (recall that $g(y)=\log\frac{\pi^d(y,x_{-1})}{\pi^d(x^d)}$ and $f^{-1}(y)=\ell R_d^{1/2}(y)$) and use the results in \cref{proof_lemma_approx5}.
Since $\Phi,\Phi',\Phi''$ are bounded, it suffices to bound all the following terms to be $o(d^{1/2})$:
\[
&\left[\frac{\dee}{\dee y}(fg-f^{-1}/2)\right]^2,\quad\frac{\dee^2}{\dee y} (fg-f^{-1}/2),\quad \exp(g) \left(\frac{\dee g}{\dee y}\right) \frac{\dee }{\dee y}(-fg-f^{-1}/2),\\
&\exp(g)\left[\frac{\dee}{\dee y}(fg-f^{-1}/2)\right]^2,\quad\exp(g)\frac{\dee^2}{\dee y} (fg-f^{-1}/2),\quad\exp(g)\left(\frac{\dee^2 g}{\dee y^2}\right),\quad \exp(g)\left(\frac{\dee g}{\dee y}\right)^2.
\]
Next, we show that most of them can be verified using Assumption \ref{assumption_A4_plus}, and the results in \cref{proof_lemma_approx5}:
	\[
	&\left[\frac{\dee}{\dee y}(fg-f^{-1}/2)\right]^2=\bigO\left[ \left(\sup_{x^d \in F_d^+}\log \pi^d(x^d) \bigO(d^{\alpha/4})+ \sup_{x^d \in F_d^+}\frac{\partial \log \pi^d}{\partial x_1}\right)^2\right]\\
	&\qquad= \bigO\left[(d^{\alpha/4}\log d + d^{\alpha/2})^2\right] =o(d^{1/2}),\\
	&\left|e^g (\frac{\dee g}{\dee y}) \frac{\dee }{\dee y}(-fg-f^{-1}/2)\right|=\bigO\left[\sup_{x^d \in F_d^+}\pi^d(x^d) \sup_{x^d \in F_d^+}\frac{\partial \log \pi^d}{\partial x_1}\left(d^{\alpha/4}\log d + d^{\alpha/2}\right)\right]\\
	&\qquad=o(d^{1/2-\alpha}d^{\alpha/2}(d^{\alpha/4}\log d + d^{\alpha/2}))=o(d^{1/2}),\\
	&\left|\exp(g)\left[\frac{\dee}{\dee y}(fg-f^{-1}/2)\right]^2\right|=o(d^{1/2-\alpha}d^{\alpha})=o(d^{1/2}),\\
	&\left|\exp(g)\left(\frac{\dee^2 g}{\dee y^2}\right)\right|=\bigO\left[\sup_{x^d \in F_d^+}\pi^d(x^d) \sup_{x^d\in F_d^+}\frac{\partial^2 \log \pi^d}{\partial x_1^2}\right]=o(d^{1/2-\alpha})\bigO(d^{\alpha})=o(d^{1/2}),\\
	&\left|\exp(g)\left(\frac{\dee g}{\dee y}\right)^2\right|=\bigO\left[\sup_{x^d \in F_d^+}\pi^d(x^d) \sup_{x^d\in F_d^+}\left(\frac{\partial \log \pi^d}{\partial x_1^2}\right)^2\right]=o(d^{1/2-\alpha})\bigO(d^{\alpha/2})^2=o(d^{1/2}).
	\]
	The only terms left are $\frac{\dee^2}{\dee y} (fg-f^{-1}/2)$ and $\exp(g)\frac{\dee^2}{\dee y} (fg-f^{-1}/2)$.
	Therefore, it suffices to show
	\[\label{tmp_eq3}
	\frac{\dee^2}{\dee y} (fg-f^{-1}/2)=\bigO(d^{\alpha}).
	\]
	Note that 
	\[
		&\frac{\dee^2}{\dee y} (fg-f^{-1}/2)
		=\frac{\dee}{\dee y}(f'g+g'f-\frac{1}{2}\dee f^{-1})\\
		&=\frac{\dee}{\dee y}\left[\frac{1}{R_d}\frac{R_d'}{R_d^{1/2}}g +\frac{1}{R_d^{1/2}}g' -\frac{1}{2}\frac{R_d'}{R_d^{1/2}} \right]\\
		&=\frac{1}{R_d}\frac{R_d'}{R_d^{1/2}}g' +\left(\frac{1}{R_d}\frac{R_d'}{R_d^{1/2}}\right)'g +\frac{1}{R_d^{1/2}}g'' + \left(\frac{1}{R_d^{1/2}}\right)'g' -\frac{1}{2}\left(\frac{R_d'}{R_d^{1/2}}\right)'.
	\]
	Note that we have shown $R_d'=o(R_d^{1/2})$ in \cref{proof_lemma_approx5}. Similarly, we also can show using Assumption \ref{assumption_A3_plus} that
	\[
	R_d''&=\frac{1}{d-1}(\sum_{j\neq 1} f_jf_j')'= \frac{1}{d-1}\sum_{j\neq 1} (f_j')^2 +\frac{1}{d-1}\sum_{j\neq 1} f_j f_j''\\
	&\le \frac{1}{d-1}\sum_{j\neq 1} (f_j')^2 + \sqrt{\frac{1}{d-1}\sum_{j\neq 1} f_j^2}\sqrt{\frac{1}{d-1}\sum_{j\neq 1} (f_j'')^2}\\
	&=\bigO(l_d/d)o((\sqrt{d/l_d})^2) + o(R_d^{1/2} \sqrt{l_d/d (\sqrt{d/l_d})^2})=o(R_d^{1/2}),
	\]
	where $f_j(x):=\frac{\partial \log \pi_{-1}(x_{-1}\,|\,x_1=x)}{\partial x_j}$.
	Therefore $R_d''=o(R_d^{1/2})$ as well. Finally, we can complete the proof by verifying \cref{tmp_eq3} using Assumption \ref{assumption_A4_plus} as follows.
	\[
	&\left|\frac{1}{R_d}\frac{R_d'}{R_d^{1/2}}g'\right|= \bigO\left(\frac{1}{R_d}\right)o(1)\bigO\left(\sup_{x^d \in F_d^+}\frac{\partial \log \pi^d}{\partial x_1}\right)= \bigO(d^{\alpha/4})o(d^{\alpha/2})=o(d^{\alpha}),\\
	&\left|\left(\frac{1}{R_d}\frac{R_d'}{R_d^{1/2}}\right)'g\right|=\bigO\left[ \frac{R_d'' R_d^{3/2}+ 3/2(R_d')^2R_d^{1/2}}{R_d^3}g\right]=\bigO\left[\frac{1}{R_d^{3/2}}(R_d''g)\right]\\
	&\qquad\qquad=\bigO(d^{\alpha/4})o(1)\bigO(d^{\alpha/2})=o(d^{\alpha}),\\
	& \left|\frac{1}{R_d^{1/2}}g''\right|=\bigO\left(\sup_{x^d \in F_d^+}\frac{\partial^2\log \pi^d}{\partial x_1^2}\right)=o(d^{\alpha}),\\
	& \left|\left(\frac{1}{R_d^{1/2}}\right)'g'\right|= \bigO\left(\frac{1}{2}\frac{1}{R_d^{3/2}}R_d'g'\right)=o(1/R_d)\bigO\left(\sup_{x^d \in F_d^+}\frac{\partial\log \pi^d}{\partial x_1}\right)=\bigO(d^{\alpha/4})o(d^{\alpha/2})=o(d^{\alpha}),\\
	&\left|\left(\frac{R_d'}{R_d^{1/2}}\right)'\right|= \left|\frac{R_d''R_d^{1/2}-\frac{1}{2}(R_d')^2\frac{1}{R_d^{1/2}}}{R_d}\right|=\bigO\left(R_d''/R_d^{1/2}\right)=o(1)=o(d^{\alpha}).
	\]

\section{Proof of \cref{thm_diffusion2}}\label{proof_thm_diffusion2}
We follow the same approach as in the proof of \cite[Proposition 3]{Roberts2016}. The idea is to follow the proof of \cref{thm_diffusion} except in the proof of \cref{tmp_eq_M_by_hatM}, we need a stronger version of \cref{lemma_approx3} to determine the sequence of ``typical sets'' $\{F_d'\}$.

Given fixed time $t$, considering the sequence of ``typical sets'' $\{F_d'\}$ defined by 
\[
F_d':=\{x^d: |R_d+S_d|\le d^{-\delta}\},
\]
where $\delta>0$ and we used $R_d$ and $S_d$ to denote $R_d^{(1)}$ and $S_d^{(1)}$ for simplicity.
We need to guarantee that when $d$ is large enough, we always have $X^d(\lfloor ds\rfloor)\in F_d\cap F_d', \forall 0\le s\le t$ and this happens for almost all starting state $X^d_1(0)=x$. That is, defining
\[
p(d,x):=\Pr(X(\lfloor ds\rfloor)\notin F_d\cap F_d', \exists 0\le s\le t \,|\, X_1^d(0)=x),
\]
letting $\pi_1$ denote the marginal stationary distribution for the first coordinate, we want to show that for any given $\epsilon>0$, as $d\to\infty$
\[
\Pr_{x\sim \pi_1}[p(d,x)\ge \epsilon, \textrm{infinite often}]=0.
\]
We prove it using Borel--Cantelli Lemma. Note that the application of Borel--Cantelli lemma is valid since we have assumed all of the processes are jointly defined on the same probability space as independent processes. First, note that 
\[
\EE_{x\sim \pi_1}[p(d,x)]=dt\Pr_{\pi^d}((F_d\cap F_d')^c)=dt\Pr_{\pi^d}(F_d^c\cup (F_d')^c)\le dt\Pr_{\pi^d}(F_d^c)+ dt\Pr_{\pi^d}((F_d')^c).
\]
For any given $\epsilon>0$, we have
\[
&\sum_{d=2}^{\infty}\Pr(p(x,d)\ge \epsilon)\le \sum_{d=2}^{\infty}\frac{\EE_{x\sim \pi_1}[p(d,x)]}{\epsilon}\\
&\le \frac{dt}{\epsilon}\sum_{d=2}^{\infty}\Pr_{\pi^d}(|R_d+S_d|> d^{-\delta})+\frac{dt}{\epsilon}\sum_{d=2}^{\infty}\Pr(X^d\notin F_d).
\]
By $\pi^d(F_d^c)=\bigO(d^{-2-\delta})$, we have $dt\sum_{d=2}^{\infty}\Pr(X^d\notin F_d)<\infty$.
Now in order to use Borel--Cantelli Lemma, the condition we need is that for some number of moments $m$ such that
\[\label{equ_moments}
\Pr_{\pi^d}(|R_d+S_d|> d^{-\delta})\le \frac{\EE|R_d+S_d|^m}{d^{-m\delta}}=d^{m\delta}\EE|R_d+S_d|^m=\bigO(d^{-2-\delta}),
\]	
which leads to
$\sum_{d=2}^{\infty}\Pr(p(x,d)\ge \epsilon)<\infty$. In order to obtain non-trivial conditions, we let $m=5$ and Assumption \ref{assumption_A2_pp} implies $\EE|R_d+S_d|^5=\bigO(d^{-2-6\delta})$.
We can then use this sequence of typical sets $\{F_d'\}$ in the proof of \cref{thm_diffusion} to replace the sequence of $\{F_d'\}$ used in \cref{lemma_approx3}. The residual proof follows the same as \cref{thm_diffusion}.

\section{Proof of \cref{prop_realistic_example}}\label{proof_prop_realistic_example}
Note that we have the number of parameters $d=n^2+n+2$ in this example. The target distribution (i.e. the posterior distribution) satisfies
\[
\begin{split}
\pi^d(x^d)&=\Pr(x^d~|~ \{Y_{ij}\}_{i,j=1}^n)\\
&\propto \frac{b^a}{\Gamma(a)}A^{-a-1}e^{-b/A} \prod_{j=1}^n\frac{1}{\sqrt{2\pi A }}e^{-\frac{(\mu_j-\nu)^2}{2A}}\prod_{i=1}^n\frac{1}{\sqrt{2\pi V }}e^{-\frac{(\theta_{ij}-\mu_j)^2}{2V}}\frac{1}{\sqrt{2\pi W}}e^{-\frac{(Y_{ij}-\theta_{ij})^2}{2W}}.
\end{split}
\]
Clearly, this model can be represented by the graphical model in \cref{graphical_model0}. It can be easily checked that the maximum number cliques any coordinate belongs to is $n+1$ and the cardinality of cliques is bounded by constant $2$, so $\sup_k |C_k|=o(d/m_d)=o(n)$. Furthermore, the target distribution clearly satisfies ``flat tail'' condition required by \cref{corollary_jeffrey} since all the conditional distributions are standard distributions. Therefore,  the first equation in \ref{assumption_A1}, the first equation in \ref{assumption_A3}, and \ref{assumption_A2} hold by \cref{corollary_jeffrey}.

Next, we verify  \ref{assumption_A5} using \cref{corollary_aaron}. Note that this model can be represented by the graphical model in \cref{graphical_model} using $K=3$ layers. In order to check the conditions in \cref{corollary_aaron}, note that
\[
\log\pi^d\propto (-a-1-\frac{n}{2})\log A -\frac{b}{A}-\frac{\sum_j (\mu_j-\nu)^2}{2A}-\frac{\sum_{i,j}(\theta_{ij}-\mu_j)^2}{2V}-\frac{\sum_{i,j}(Y_{ij}-\theta_{ij})^2}{2W}.
\]
Observing that, under $X^d=(\nu,A,\{\mu_j\}_{j=1}^n, \{\theta_{ij}\}_{i,j=1}^n)\sim \pi^d$, we have
\[
\theta_{ij}~|~Y_{ij},\mu_j\quad &\sim^{\textrm{indep.}} \mathcal{N}\left(\frac{W\mu_j+VY_{ij}}{W+V},\frac{VW}{W+V}\right),\quad i,j\in\{1,\dots, n\},\\
\mu_j~|~\sum_i \theta_{ij}, \nu, A &\sim^{\textrm{indep.}} \mathcal{N}\left(\frac{\sum_i A\theta_{ij}+V\nu}{nA+V},\frac{AV}{nA+V}\right), \quad i\in \{1,\dots, n\},\\
\nu~|~\bar{\mu}, A &\sim \mathcal{N}\left(\bar{\mu},\frac{A}{n}\right),\\
A~|~\{\mu_j\},\nu &\sim \textbf{IG}\left(a+\frac{n}{2}, b+\frac{1}{2}\sum_j(\mu_j-\nu)^2\right).
\]	
Therefore, we have
\[
\left|\frac{\partial \log \pi^d}{\partial A}\right|=\left|\frac{b+\frac{1}{2}\sum_j(\mu_j-\nu)^2}{A^2}-\frac{a+1+\frac{n}{2}}{A}\right|=\bigO_{\Pr}(d^{1/2}).
\]
since $\frac{a+1+n/2}{A}\to_{\Pr} \frac{a+1+n/2}{A_0}=\bigO(d^{1/2})$ and $\sum_j(\mu_j-\nu)^2
\to_{\Pr}\sum_j(\mu_j-\bar{\mu})^2+\frac{A_0}{n}=\bigO_{\Pr}(d^{1/2})$.
Other coordinates can also be verified, which are shown as follows.
\[
\left(\frac{\partial \log \pi^d}{\partial \nu}\right)^2&=\left(\frac{n(\bar{\mu}-\nu)}{A}\right)^2=\bigO_{\Pr}\left(\frac{n}{A}\right)=\bigO_{\Pr}(d/n),\\
\left(\frac{\partial \log \pi^d}{\partial \mu_j}\right)^2&=\left(\frac{\sum_i (\theta_{ij}-\mu_j)}{V} -\frac{\mu_j-\nu}{A}\right)^2=(nA+V)^2\left(\frac{A\sum_i \theta_{ij}+V\nu}{nA+V}-\mu_j\right)^2\\
&=\bigO_{\Pr}\left[ (nA+V)^2 \frac{AV}{nA+V}\right]=\bigO_{\Pr}(d/n),\\
\left(\frac{\partial \log \pi^d}{\partial \theta_{ij}}\right)^2&=\left(\frac{Y_{ij}-\theta_{ij}}{V}-\frac{\theta_{ij}-\mu_j}{W}\right)^2=(W+V)^2\left(\frac{VY_{ij}+W\mu_j}{W+V}-\theta_{ij}\right)^2=\bigO_{\Pr}(d/n^2).\\
\]
Therefore,  \ref{assumption_A5} holds by \cref{corollary_aaron}.  Finally, all the other conditions in \ref{assumption_A1}, \ref{assumption_A3}, and \ref{assumption_A4} can be verified in a similar way as in \cref{subsection_verify} for \cref{example_toy}. 
\end{document}